\documentclass[acmsmall]{acmart}
\renewcommand\footnotetextcopyrightpermission[1]{} 

\def\BibTeX{{\rm B\kern-.05em{\sc i\kern-.025em b}\kern-.08emT\kern-.1667em\lower.7ex\hbox{E}\kern-.125emX}}

\usepackage{etex}
\usepackage{wrapfig}
\usepackage{graphicx}
\usepackage{amssymb,amsmath,amsthm}
\usepackage[mathscr]{eucal}
\usepackage{verbatim}
\usepackage{comment}
\usepackage{paralist}
\usepackage{mathtools}
\usepackage{amsfonts}
\usepackage{color}
\usepackage{graphicx}
\usepackage{tikz}
\usepackage{xspace} 
\usepackage{pgf}
\usetikzlibrary{arrows,automata}
\usepackage[usenames,dvipsnames]{pstricks}
\usepackage{epsfig}
\usepackage{pst-grad} 
\usepackage{pst-plot} 
\usepackage{pstricks-add}
\usepackage[linesnumbered,ruled]{algorithm2e}
\usepackage{algorithmicx}

\newtheorem{theorem}{Theorem}
\newtheorem{proposition}{Proposition}
\newtheorem{lemma}{Lemma}
\newtheorem{corollary}{Corollary}
\newtheorem*{remark}{Remark}

\newtheorem{assumption}{Assumption}
\theoremstyle{definition}
\newtheorem{definition}{Definition}
\newtheorem{example}{Example}

\newcommand{\B}{\boldsymbol{\beta}}

\def\reals{\mathbb{R}}
\def\nats{\mathbb{N}}
\def\set#1{{\{ #1 \}}}

\def\Prob{{\mathrm{Prob}}}
\def\realpart{{\mathrm{Re}}}
\def\M{{\mathcal{M}}}
\def\N{{\mathcal{N}}}
\def\I{{\mathbb{I}}}
\def\good{\mathbf{good}}
\def\bad{\mathbf{bad}}
\def\bisim{{\cong}}

\begin{document}
\acmJournal{TOMPECS}
\acmDOI{}
\acmISBN{}
\acmPrice{}
\acmArticleSeq{} 
\acmVolume{}
\acmNumber{}
\acmSubmissionID{}
\title[A Lyapunov Approach for Time Bounded Reachability]{A Lyapunov Approach for Time Bounded Reachability\\ of CTMCs and CTMDPs}

\author{Mahmoud Salamati}
\email{msalamati@mpi-sws.org}
\affiliation{
	\institution{MPI-SWS}
	\city{Kaiserslautern}
	\country{Germany}
}
\author{Sadegh Soudjani}
\affiliation{
	\institution{Newcastle University}
	\city{Newcastle upon Tyne}
	\country{United Kingdom}}
\email{Sadegh.Soudjani@newcastle.ac.uk}

\author{Rupak Majumdar}
\email{rupak@mpi-sws.org}
\affiliation{
	\institution{MPI-SWS}
	\city{Kaiserslautern}
	\country{Germany}
}
\renewcommand{\shortauthors}{M. Salamati et al.}
\begin{abstract}
Time bounded reachability is a fundamental
problem in model checking continuous-time Markov chains (CTMCs) 
and Markov decision processes (CTMDPs) for specifications in continuous stochastic logics.
It can be computed by numerically solving a characteristic linear dynamical system but the procedure
is computationally expensive.
We take a control-theoretic approach and 
propose a reduction technique that finds another dynamical system of lower dimension (number of variables),
such that numerically solving the reduced dynamical system provides an approximation to the 
solution of the original system with guaranteed error bounds.
Our technique generalises lumpability (or probabilistic bisimulation) to a quantitative setting.
Our main result is a Lyapunov function characterisation of the difference in the trajectories of the two dynamics
that depends on the initial mismatch and exponentially decreases over time.
In particular, the Lyapunov function enables us to compute an error bound between the two dynamics as well as 
a convergence rate.
Finally, we show that the search for the reduced dynamics can be computed 
in polynomial time using a Schur decomposition of the transition matrix.
This enables us to efficiently solve the reduced dynamical system by computing the exponential of an
upper-triangular matrix characterising the reduced dynamics.
For CTMDPs, we generalise our approach using piecewise quadratic Lyapunov functions for switched affine dynamical
systems. We synthesise a policy for the CTMDP via its reduced-order switched system 
that guarantees the time bounded reachability probability lies above a threshold. 
We provide error bounds that depend on the minimum dwell time of the policy.
We demonstrate the technique on examples from queueing networks, for which lumpability does not
produce any state space reduction but our technique synthesises policies using reduced version of the model.
\end{abstract}



\maketitle

\section{Introduction}
\label{sec:intro}

Continuous-time Markov chains (CTMCs) and Markov decision processes (CTMDPs)
play a central role in the modelling and analysis of performance and dependability
properties of probabilistic systems evolving in real time.
A CTMC combines probabilistic behaviour with real time:
it defines a transition system on a set of states, where the transition between two
states is delayed according to an exponential distribution.
Any state of the system may have multiple possible next states, each with an associated exponentially-distributed delay. The next state is chosen according to a race condition among these delays.
A CTMDP extends a CTMC by introducing non-deterministic choice among a set of possible
actions.
Both CTMCs and CTMDPs have been used in a large variety of applications ---from biology to
finance.

A fundamental problem in the analysis of CTMCs and CTMDPs is \emph{time bounded reachability}:
given a CTMC, a set of states, a time bound $T$, and a real value $\theta\in[0,1]$, it asks whether 
the probability
of reaching the set of states within time $T$ is at least $\theta$.
In CTMDPs we are interested in synthesising a policy that resolves non-determinism for satisfying this requirement.
Time bounded reachability is the core technical problem for model checking stochastic
temporal logics such as Continuous Stochastic Logic \cite{Aziz:2000,Baier:2003}, and 
having efficient implementations of time bounded reachability is crucial to scaling formal
analysis of CTMCs and CTMDPs.

Existing approaches to the time bounded reachability problem are based on discretisation or uniformisation,
and in practice, are expensive computational procedures, especially as the time bound increases.
The standard state-space reduction technique is probabilistic bisimulation \cite{utm86,Larsen:1991,Buchholz:1999,Baier:2003}: 
a probabilistic bisimulation is an equivalence relation on the states that allows ``lumping'' together
the equivalence classes without changing the value of time bounded reachability properties, or indeed 
of any CSL property \cite{Baier:2003}.
Unfortunately, probabilistic bisimulation is a strong notion and small perturbations to the transition
rates can change the relation drastically.
Thus, in practice, it is often of limited use.

In this paper, we take a control-theoretic view to state space reductions of CTMCs and CTMDPs.
Our starting point is that the forward Chapman-Kolmogorov equations characterising time bounded reachability
define a linear dynamical system for CTMCs and a switched affine dynamical system for CTMDPs;
moreover, one can transform the problem so that the dynamics is stable.
Our first observation is a generalisation of probabilistic bisimulation to a quantitative setting.
We show that probabilistic bisimulation can be viewed as a projection matrix that relates the original dynamical
system with its bisimulation reduction. 
We then relax bisimulation to a quantitative notion, using a \emph{generalised projection} operation between two linear systems.

\paragraph{CTMCs.}
A generalised projection does not maintain a linear relationship between the original and the reduced linear systems.
However, our second result shows how the difference between the states of the two linear dynamical systems can be bounded
by an exponentially decreasing function of time.
The key to this result is finding an appropriate Lyapunov function for the difference between the two dynamics,
which demonstrates an exponential convergence over time.
We focus the presentation of the paper on irreducible CTMCs (i.e., those with the property that it is possible with some positive probability to get from any state to any other state in finite time) and show that the search for a suitable Lyapunov function can be reduced to a system of matrix inequalities,
which have a simple solution. This leads to an error bound of the form $L_0 e^{-\kappa t}$,
where $L_0$ depends on the matrices defining the dynamics, and $\kappa$ is related to the eigenvalues of the dynamics.
Clearly, the error goes to zero exponentially as $t\rightarrow \infty$.
Hence, by solving the reduced linear system, one can approximate the time bounded reachability probability
in the original system, with a bound on the error that converges to zero as a function of the reachability horizon.
For reducible CTMCs (i.e., those that are not irreducible), we show that the same approach is applicable by preprocessing the 
structure of CTMC and eliminating those bottom strongly connected components that do not influence the reachability probability. 

The Lyapunov approach suggests a systematic procedure to reduce the state space of a CTMC.
If the original dynamical system has dimension $m$, we show, using Schur decomposition, that we can
compute an $r$-dimensional linear system for each $r\leq m$ as well as a Lyapunov-based bound on the
error between the dynamics.
Thus, for a given tolerance $\varepsilon$, one can iterate this procedure to find an appropriate $r$.
This $r$-dimensional system can be solved using existing techniques, e.g., computing the exponential of upper-triangular matrices.

\paragraph{CTMDPs.}
For CTMDPs, we generalise the approach for CTMCs using Lyapunov stability theorems for switched systems.
Once again, the objective is to use multiple Lyapunov functions as a way to demonstrate stability,
and derive an error bound from the multiple Lyapunov functions.
For this we construct a piecewise quadratic Lyapunov function for a switched affine dynamical system. 
Then we synthesise a policy for the CTMDP via its reduced-order switched system in order to have time 
bounded reachability probability above a threshold. 
We provide error bounds that depend on the minimum dwell time of the policy.

The notion of \emph{behavioural pseudometrics} on stochastic systems has been studied extensively \cite{Bacci1:2015,Desh:2004} as a quantitative measure of dissimilarity between states, but mainly for discrete time Markov models and mostly
for providing an upper bound on the difference between \emph{all} formulas in a logic;
by necessity, this makes the distance too pessimistic for a single property.
In contrast, our approach considers a notion of distance for a specific time-bounded reachability property,
and provides a time-varying error bound.

We have implemented our state space reduction approach and evaluated its performance on a queueing
system benchmark. Fixing time horizon and error bound, our reduction algorithm computes a reduced order system, the analysis of which requires a significantly less computational effort.
 We show that, as the time horizon increases, we get significant reductions in the dimension
of the linear system while providing tight bounds on the quality of the approximation.

 \smallskip
A subset of the results of this paper has been presented in \cite{QEST18}. The current paper improves \cite{QEST18} in the following directions.
First,
we have exploited the properties of the reduced order system and proposed a symbolic computation that reduces the computational time empirically by two orders of magnitude.
Second,
we have provided a systematic approach for quantifying the reduction error bound on CTMDPs. 
The formulated error in \cite{QEST18} requires solving a min-max optimisation problem, while the new approach does not need such optimisations.
Third,
we have studied the error induced on the reachability probabilities due to the perturbations on the parameters of the model. 
Finally, while the presentation of the paper is focused on irreducible models, we have shown that the results are applicable to models that are not irreducible.
We have also provided the proofs of all the statements and algorithmic procedures for performing the computations.

\section{Continuous-Time Markov Chains}
\label{sec:preliminaries}

\begin{definition}
\label{def:CTMC}
A \emph{continuous-time Markov chain (CTMC)} $\mathcal M = (S_{\mathcal{M}},R,\alpha)$ 
consists of a finite set $S_{\mathcal{M}}=\{1,2,\cdots,\mathfrak n\}$ of states for some positive natural number $\mathfrak n$,
a rate matrix $S_{\mathcal{M}}\times S_{\mathcal{M}}\rightarrow \mathbb R_{\ge 0}$, and
an initial probability distribution $\alpha:S_{\mathcal{M}}\rightarrow[0,1]$ satisfying
$\sum_{s\in S_{\mathcal{M}}}\alpha(s) = 1$.
\end{definition}

Intuitively, $R(s,s')>0$ indicates that a transition from $s$ to $s'$ is possible and that the timing of the transition is exponentially distributed with rate $R(s,s')$. If there are several states $s'$ such that $R(s,s') > 0$, the chain can transition to the state with the minimum time. This property is known as the race condition between exponentially distributed transition times.
Denote the total rate of taking an outgoing transition from state $s$ by 
$E(s)=\sum_{s'\in S_{\mathcal{M}}}{R(s,s')}$.
A transition from a state $s$ into $s'$ wins within time $t$ with probability
\begin{equation*}
\mathbf{P}(s,s',t)=\frac{R(s,s')}{E(s)}.(1-e^{-E(s)t}).
\end{equation*}
Intuitively, $1-e^{-E(s)t}$ is the probability of taking an outgoing transition at $s$ within time $t$ 
(exponentially distributed with rate $E(s)$) and 
$R(s,s')/E(s)$ is the probability of taking transition to $s'$ among possible next states at $s$. 
Thus, the probability of moving from $s$ to $s'$ in one transition, written $\mathbf{P}(s,s')$ is $\frac{R(s,s')}{E(s)}$.
A state $s\in S_{\mathcal{M}}$ is called \emph{absorbing} if and only if $R(s,s') = 0$ for all $s'\in S_{\mathcal{M}}$.
For an absorbing state, we have $E(s) = 0$ and no transitions are enabled.

A right continuous-step function $\rho : \reals_{\geq 0}\rightarrow S_{\mathcal{M}}$
is called an infinite path. Such a function is piece-wise constant with countable number of discontinuity points and is right-continuous.
For a given infinite path $\rho$ and $i\in\nats$, we denote by $\rho_S[i]$
the state at the $(i +1)$-st step, and by $\rho_T [i]$ the time spent
at $\rho_S[i]$, i.e., the length of the step segment starting with
$\rho_S [i]$. Note that the definition of infinite paths allows for the case of finite number of steps and hence, $\rho_S[i]$ and $\rho_T[i]$ may be unbounded for some $i$.
Let $\Pi_{\mathcal{M}}$ denote the set of all infinite paths, and
$\Pi_{\mathcal{M}}(s)$ denote the subset of those paths starting from $s\in S_{\mathcal{M}}$.
Let $I_0,\ldots, I_{k-1}$ be nonempty intervals in $\reals_{\geq 0}$. 
The cylinder set $\mathit{Cyl}(s_0, I_0, s_1, I_1,\ldots, s_{k-1}, I_{k-1}, s_k)$ is defined
by:
\[
\set{\rho \in \Pi_{\mathcal{M}}\mid 
\forall 0\leq i\leq  k\,.\, \rho_S [i]
  = s_i \wedge \forall 0 \le  i < k\,.\, \rho_T [i] \in I_i }.
\]
Let $\mathcal{F}(\Pi_{\mathcal{M}})$ denote the smallest $\sigma$-algebra on $\Pi_{\mathcal{M}}$ containing
all cylinder sets. 
The probability measure $\Prob_\alpha$ on $\mathcal{F}(\Pi_\M)$ is the unique measure defined by induction
on $k$ as
\begin{align*}
\Prob_\alpha(& \mathit{Cyl} (s_0, I_0, \ldots, s_k, [a,b], s')) = \\
& \Prob_\alpha(\mathit{Cyl}(s_0,I_0,\ldots, s_k))\cdot 
 \mathbf{P}(s_k, s')(e^{-E(s_k)a} - e^{-E(s_k)b}).
\end{align*}
The \emph{transient state probability}, written $\bar\pi^\M_\alpha(t)$, is defined as a row vector 
indexed by $S_{\mathcal M}$ 
with the value $\Prob_\alpha\set{\rho\mid \rho(t) = s'}$ for each $s'\in S_{\mathcal M}$.
The transient probabilities of $\M$ are characterised by the forward Chapman-Kolmogorov differential equation \cite{Boyd:1994}, 
which is the system of linear differential equations 
\begin{equation}
\label{eq:Kolmog}
\frac{d}{dt} \bar{\pi}^{\mathcal M}_\alpha(t) = 
\bar\pi^{\mathcal M}_\alpha(t) \bar{\mathbf Q}, \quad \bar{\pi}^{\mathcal M}_\alpha(0) = \alpha. 
\end{equation}
where $\bar{\mathbf Q}$ is the \emph{infinitesimal generator} matrix of 
$\M$ defined as $\bar{\mathbf Q} = R - diag_s(E(s))$. 
Note that $\sum_{s'}{\bar{\mathbf Q}(s,s')}=0$ for any $s\in{S_{\mathcal{M}}}$.
The solution $\pi^{\mathcal M}_s(t)(s')$ indicates the probability 
that $\mathcal M$ starts at initial state $s$ and is at state $s'$ at time $t$. Therefore,
\begin{equation}
\label{eq:diff}
\frac{d}{dt} \bar{\pi}^{\mathcal M}_s(t) = 
\bar\pi^{\mathcal M}_s(t) \bar{\mathbf Q},\quad \bar\pi^\M_s(0) = \mathbf{1}(s)
\end{equation}
where $ \bar{\pi}^{\mathcal M}_s(t)\in \reals^{|S_{\mathcal{M}}|}$ is a row vector containing 
transient state probabilities ranging over all states in $S_{\mathcal{M}}$.
Using an ordering of the states of $S_{\mathcal M}$, we 
equate a row vector in $\reals^{|S_{\mathcal M}|}$ with a function in $\reals^{S_{\mathcal M}}$ from
$S_{\mathcal M}$ to reals. The initial value of differential equation \eqref{eq:diff} is a
vector indicating the initial probability distribution that assigns
the entire probability mass to the state $s$, that is, 
$\bar{\pi}_s^{\mathcal M}(0) = \mathbf 1(s)$, a vector that assigns $s$ to one and every other state to zero.

Let $\mathcal{M} = (S\uplus\set{\good,\bad}, R, \alpha)$ be a CTMC
with two states $\good$ and $\bad$.
Let $|S|=m$ and let $T \in\reals_{\geq 0}$ be a time bound.
We write $\Prob^\M(\mathbf 1(s), T) = \bar{\pi}^\M_s(T)(\good)$. 
The \emph{time-bounded reachability problem} asks to compute this probability. 
Note that, for all $T$, we have $\Prob^\M(\mathbf 1(\good), T) = 1$ and $\Prob^\M(\mathbf 1(\bad), T) = 0$.
In general, we are interested in finding the probability for a given subset $S_0\subseteq S$ of states.
Defining $|S_0| = m_0$,
we denote solution to this problem as a $m_0\times 1$ vector $\Prob^\M(C, T)$, 
where $C$ is a $m_0\times (m+2)$ matrix with $|m_0|$ ones on its main diagonal, corresponding
to the states in $S_0$.
If $S_0 = S_\M$, then $C$ is the $(m+2)\times (m+2)$ identity matrix.
Each element of the vector $\Prob^\M(C, T)$ is the value of $\Prob^\M(\mathbf 1(s), T)$ for the respective state $s\in S_0$.

\section{Time-Bounded Reachability on CTMCs}
\label{sec:CTMCs}

\subsection{From Reachability to Linear Dynamical Systems}

Let $\M = (S\uplus \set{\good, \bad}, R, \alpha)$ be a CTMC, with $|S|
= m$, and two states $\good$ and $\bad$.
If these two states have outgoing transitions, we make them absorbing and denote the resulted infinitesimal generator by $Q$. The solution to the time-bounded reachability problem for a projection matrix $C$
can be obtained by rewriting \eqref{eq:diff} as:
\begin{align}
& \frac{d}{dt} Z(t) = Q Z(t),\quad  Z(0) = \mathbf 1(\good),\nonumber\\
& \Prob^\M(C, t)=C Z(t)
\label{eq:differnetial2}
\end{align}
where $Z(t) \in \reals^{m+2}$ is a column vector with elements $Z_i(t)=\Prob^\M(\mathbf{1}(s_i), t)$. 
Notice that in this formulation, we have let time ``run backward'':
we start with a initial vector which is zero except for corresponding element to the
state $\good$ and compute ``backward'' up to the time $T$.
By reordering states, if necessary, the generator matrix $Q$ in
\eqref{eq:differnetial2} can be written as:
\begin{equation}
\label{eq:Q_part}
Q=
\begin{bmatrix}
A & \vdots & \boldsymbol{\chi} &\vdots & \B\\
\dots & \dots & \dots &\dots & \dots\\
\mathbf{0} &\vdots & \mathbf 0 &\vdots & \mathbf 0
\end{bmatrix}
\end{equation}
with $A\in\mathbb R^{m\times m}$, $\boldsymbol{\chi}\in\mathbb
R^{m\times 1}$, and $\boldsymbol{\beta}\in\mathbb R^{m\times 1}$. Vectors $\boldsymbol{\chi}$ and $\boldsymbol{\beta}$ contain the rates corresponding to the incoming transitions to the states $\bad$ and $\good$, respectively.
With this reordering of the states, it is obvious that in
\eqref{eq:differnetial2},  $Z(t)(\bad) = 0$ and $Z(t)(\good) = 1$, thus we assume states $\good$ and $\bad$ are not included in $C$.
We write $Z_{S}(t)$ for the vector (in $\reals^m$) restricting $Z$ to states in $S$.
These variables should satisfy
\begin{align}
& \frac{d}{dt} Z_{S}(t) = A Z_{S}(t) + \boldsymbol{\beta},\quad  Z_{S}(0) = 0,\nonumber\\
& \Prob^{\mathcal M}(C_S,t) =C_{S} Z_{S}(t),
\label{eq:differnetial3}
\end{align}
where $C_S\in \reals^{m_0 \times m}$ is the matrix obtained by omitting the last 
two columns of $C$.

Equation \eqref{eq:differnetial3} can be seen as model of a linear
dynamical system with unit input.
Our aim here is to compute an approximate solution of \eqref{eq:differnetial3} using
reduction techniques 
from control theory while providing guarantees on the accuracy of the
computation and to interpret the solution as the probability for time
bounded reachability.  

Let $\gamma:=max_{i=1:m}|a_{ii}|$, the maximal diagonal element of $A$, and define matrix $H$ as:
\begin{equation}
\label{eq:unif}
H=\frac{A}{\gamma}+\mathbb I_m,
\end{equation}
where $\mathbb I_m$ is the $m\times m$ identity matrix. 
In the following, we fix the following assumption. 
\begin{assumption}
	\label{ass:invert}
	$H$ is an irreducible matrix, i.e., its associated directed
        graph is strongly connected.
	Moreover, $\boldsymbol{\beta} + \boldsymbol{\chi} \neq 0$.
        That is, either $\good$ or $\bad$ is reachable in one step from some
        state in $S$.
\end{assumption}

\begin{remark}
	The above assumption is ``WLOG.''
	First, if there is no edge from $S$ to $\good$ or $\bad$, the problem is trivial. 
	Second, the general case, when $H$ is not irreducible can be reduced to the assumption in polynomial time (see Appendix~\ref{sec:CTMC Red}). 
	Thus, the assumption restricts attention to the core technical problem. Throughout the rest of the paper, we only consider models for which the above assumption holds.
\end{remark}

Recall that a matrix $A$ is stable if every eigenvalue of $A$ has negative real part.
The \emph{spectral radius} of a matrix is the largest absolute value of its eigenvalues.
We also denote the real part of the eigenvalues of a complex number by $\realpart(\cdot)$.

\begin{proposition}
\label{prop:stable}
	Assumption~\ref{ass:invert} implies that matrix $A$ is invertible and stable.
\end{proposition}
\begin{proof}
Due to the definition of $H$ in \eqref{eq:unif}, we have $\lambda(H) = 1 + \lambda(A)/\gamma$, where $\lambda(\cdot)$ denotes the eigenvalues of a matrix.
We use $\varrho$ for the spectral radius of $H$.  
For irreducible matrix $H$, the Perron-Frobenius theorem implies that $\varrho$ is positive and it is a simple eigenvalue of $H$. 
There are left eigenvectors associated with eigenvalue $\varrho$ such that their entries are all positive. 
Without loss of generality, we denote one of these left eigenvectors by $\nu$  
that is normalised such that sum of its entries is equal to one.
The aim is to show that $\varrho<1$.
Since the sum of every row of $H$ is less than or equal to one, $\varrho$ cannot be greater than one. 
The following reasoning shows that $\varrho=1$ gives a contradiction.
Define a diagonal matrix
$$\Delta := \mathit{diag}(\boldsymbol{\chi} + \boldsymbol{\beta})/\gamma$$
and let $\tilde{H} := H+\Delta$. 
This matrix is a row stochastic matrix and is irreducible. 
Then it can be considered as an irreducible probability transition matrix of a discrete-time Markov chain. 
Note that
\begin{equation*}
\nu\tilde{H} = \nu (H+\Delta) = \varrho \nu + \nu \Delta = \nu + \nu \Delta.
\end{equation*}
We can show by induction that the following inequality
\begin{equation}
\label{eq:induc}
\nu\tilde{H} ^k\ge \nu +\nu\Delta+\nu\Delta^2+\ldots+\nu\Delta^k
\end{equation}
holds element-wise for all $k\in\mathbb N$.
This can be seen using the inductive step
\begin{equation*}
\nu\tilde{H}^{k+1} = (\nu\tilde{H}^k)\tilde{H} \ge  (\nu + \nu \Delta +\ldots+\nu\Delta^k)(H+\Delta) \ge \nu + \nu \Delta + \nu \Delta^2 +\ldots \nu\Delta^{k+1}.
\end{equation*}
The last inequality is true since all the additional terms in its left-hand side have non-negative entries (all elements of $H,\Delta,\nu$ are non-negative).

Taking the sum of all entries of both sides of \eqref{eq:induc}, we get
$$\sum_i \nu_i \ge \sum_i \nu_i(1+\Delta_{ii}+\Delta_{ii}^2+\ldots+\Delta_{ii}^k),$$
which is a contradiction since at least one diagonal element of $\Delta$ is positive. 
Then we have $\varrho<1$, which results in $\realpart(\lambda(A))<0$ due to the relation $\lambda(H) = 1 + \lambda(A)/\gamma$. Therefore, $A$ is stable and invertible.
\end{proof}
Since the input to \eqref{eq:differnetial3} is fixed, we try to
transform it to a set of differential equations without input but with
initial value.
Let us take a transformation that translates $Z_{S}(t)$ by the offset vector $A^{-1}\boldsymbol{\beta}$:
\begin{equation}
\label{eq:translation}
X(t) := Z_{S}(t) + A^{-1}\boldsymbol{\beta}. 
\end{equation}
The evolution of $X(\cdot)$ is: 
\begin{align}
& \frac{d}{dt} X(t) = A X(t),\quad  X(0) = A^{-1}\B,\nonumber\\
& \Prob^{\mathcal M}(C_S,t) = C_S X(t)+d. 
\label{eq:Xp}
\end{align}
where $d = - C_SA^{-1}\B$.
The \emph{dimension} (number of variables) of dynamical system \eqref{eq:Xp} is $m$, the size of the state space $S$.

\begin{remark}
Under Assumption~\ref{ass:invert}, the solution of infinite horizon reachability problem is $- A^{-1}\B$, which can be computed efficiently as the solution of a system of linear equations.
Elements of $X(t)$ defined in \eqref{eq:translation} contain the values of finite-horizon reachability.
\end{remark}

In the following, we show how the solution of this dynamical
system can be approximated by a dynamical system of lower dimension.
Our approach relies on stability property of matrix $A$, and gives an upper bound on the approximation error that converges exponentially to zero as a function of time.
Thus our approach is beneficial for long time horizons when previous techniques fail to provide tight bounds.

\subsection{Bisimulation and Projections}
Probabilistic bisimulation or lumpability is a classical technique to
reduce the size of the state space of a CTMC \cite{utm86,Larsen:1991,Buchholz:1999,Baier:2003}.
For CTMC $\M = (S_{\mathcal{M}},R,\alpha)$ with space $S_{\mathcal M} = S\uplus\set{\good,\bad}$, a bisimulation on $\mathcal M$
is an equivalence relation 
$\bisim$ on $S_{\mathcal{M}}$ such that $\good$ and $\bad$ are singleton equivalence
classes and for any two states $s_1, s_2\in S$,
$s_1 \bisim s_2$ implies
$R(s_1,\Theta)=R(s_2,\Theta)$ for every equivalence class $\Theta$ of $\bisim$, where $R(s,\Theta):=\sum_{s'\in \Theta}R(s,s')$.
Given a bisimulation relation $\bisim$ on $\mathcal M$, we can construct a CTMC $\bar{\mathcal M} = (S_{\bar{\mathcal M}},\bar R,\bar\alpha)$ of smaller size
such that probabilities are preserved over paths of $\mathcal M$ and $\bar{\mathcal M}$.
In particular, $s_1 \bisim s_2$, implies that
\begin{equation*}
\Prob^{\mathcal M}(\mathbf{1}(s_1),t) = Prob^{\bar{\mathcal
    M}}(\mathbf{1}(s_2),t),\quad \forall t\in\mathbb R_{\ge 0}.
\end{equation*}
The CTMC $\bar{\mathcal M}$ has the quotient state space
$\set{[s]_\bisim \mid s\in S}\uplus \set{\good, \bad},$ where $[s]_\bisim$ is the equivalence class
of $s\in S$,
rate function $\bar R([s]_{\bisim},\Theta) = R(s,\Theta)$ for any $\Theta \in S_{\bar{\mathcal M}}$, and
initial distribution $\bar\alpha([s]_{\bisim}) = \sum_{s'\in[s]_{\bisim}}\alpha(s')$.

We now show how the differential equation \eqref{eq:Xp} for $\M$ and $\bar{\M}$
relate.
Assume that the state space of $\bar{\M}$ is $\bar{S} \cup\set{\good,\bad}$,
where $|\bar{S}| = r$. 
We have
\begin{align}
& \frac{d}{dt}\bar X(t) = \bar A \bar X(t),\quad  \bar X(0) = \bar A^{-1}\bar{\B},\nonumber\\
& Prob^{\bar{\mathcal M}}(\bar C_S,t) = d +  \bar C_S \bar X(t),
\label{eq:Xr}
\end{align}
where $\bar A$ and $\bar{\B}$ are computed
similarly to that of $\M$ according to 
the generator matrix of $\bar{\M}$.
Note that $\bar{A}$ is an $r\times r$ matrix.
Matrix $\bar C_S$ is $m_0\times r$ constructed according to $S_0$, with $|S_0|$ ones corresponding to the quotient states $\set{[s]_\bisim \mid s\in S_0}$.
We now define a \emph{projection matrix} $P_\bisim\in\mathbb R^{m\times r}$
as 
$P_\bisim(i,j) = 1$ if $s_i \in [j]$, i.e., $s_i$ belongs to the
equivalence class $[j] \in \bar{S}$,
and zero otherwise.
This projection satisfies 
 $C_S P_\bisim = \bar{C}_S$, and together with
the definition of $\bisim$ implies the following proposition.

\begin{proposition}
\label{prop:projection}
For every bisimulation $\bisim$,
the projection matrix $P_\bisim$ satisfies the following
\begin{equation}
\label{eq:matrix_transf}
AP_\bisim = P_\bisim \bar A,\quad \B = P_\bisim\bar{\B}.
\end{equation}
Conversely, every projection matrix satisfying
\eqref{eq:matrix_transf} defines a bisimulation relation.
In particular, 
\begin{equation}
\label{eq:dynamics-related}
X(t) = P_\bisim \bar X(t),\quad \forall t\in\mathbb R_{\ge 0}. 
\end{equation}
\end{proposition}

\begin{example}
	As an example, consider the CTMC in Fig.~\ref{example_full}
        with $\Lambda_{31}=0$ and $\Lambda_{42}=1$ without any state $\bad$, and assume first that $\varepsilon_{ij} = 0$ for all $i,j$.
	We are interested in computing the probability of
        reaching state $\good$, which is made absorbing by removing
        its outgoing links. 
        It is easy to see that the bisimulation classes are $\set{s_1,
          s_2}$, $\set{s_3, s_4}$, and $\set{\good}$.
        The bisimulation reduction and the corresponding projection matrix $P_\bisim$ are shown on the right-hand side.
               The differential equation for the reduced CTMC has dimension 2.
\end{example}

Unfortunately, as is well known, bisimulation is a strong condition,
and small perturbations in the rates can cause two states to not be bisimilar.
Consider a perturbed version of the CTMC by setting $\varepsilon_{23} = -\varepsilon_{13} = 0.05$, which will give the following generator matrix:
	\begin{equation*}
	Q=
	\begin{bmatrix}
	-3.95 & 0 & 1.95 & 0 & 2\\
	0 & -4.05 & 1.05 & 1 & 2\\
	0 & 1 & -1 & 0 & 0\\
	0 & 1 & 0 & -1 & 0\\
	0 & 0 & 0 & 0 & 0
	\end{bmatrix}.
	\end{equation*}	
Here, $\varepsilon_{ij} \neq 0$ for some transitions, and the CTMC on the right-hand side of Fig.~\ref{example_full} is not a bisimulation reduction.
Let us also consider a perturbed version of  the CTMC on the right-hand side of Fig.~\ref{example_full} with the generator matrix
\begin{equation*}
	Q_r=
	\begin{bmatrix}
	-4.05 & 2.05 & 2\\
	1 &-1 & 0\\
	0 & 0 & 0
	\end{bmatrix}.
	\end{equation*}
Clearly, these two perturbed CTMCs are not bisimilar according to the usual definition of bisimulation relation, but the following real matrix
\begin{equation*}
	P=
	\begin{bmatrix}
	\frac{390}{469} & \frac{39}{469} & \frac{40}{469}\\
	1 & 0 & 0\\
	0 & 1 & 0\\
	0 & 1  & 0\\
	0 & 0  & 1
	\end{bmatrix},
	\end{equation*}
satisfies the equality $QP=PQ_r$.
Note that $P$ is no longer a projection matrix but has entries in
$[0,1]$, which sum up to $1$ for each row.
This particular $P$ satisfies $AP=P\bar{A}$ but not $\B=P\bar{\B}$ (see \eqref{eq:matrix_transf}).
Thus the original dynamics of $X(t)$ and their lower-dimensional version $\bar X(t)$, reduced with $P$, do not satisfy the equality
\eqref{eq:dynamics-related}.

However, since $A$ is a stable matrix, we expect the trajectories of
the original and the reduced dynamics to converge, that is, the 
error between the trajectories to go to zero as time goes to infinity.
In the next section, we generalise projection matrices as above, and 
formalise this intuition.

	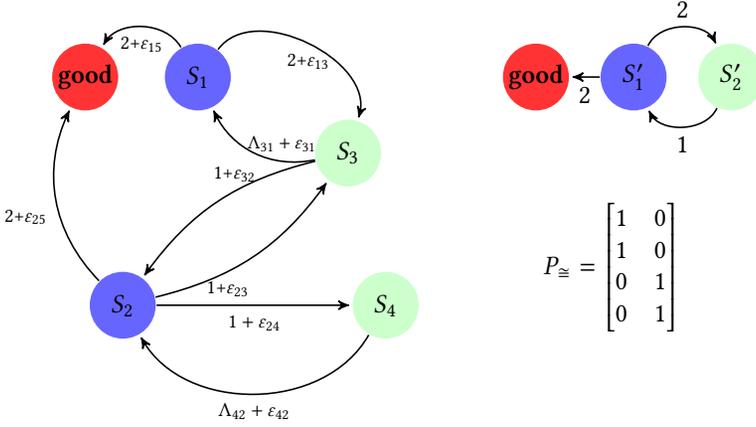
\begin{figure}[t]
		\begin{tikzpicture}[->,>=stealth',shorten >=1pt,auto,semithick]
		\tikzstyle{every state}=[draw=none,text=black,inner sep=0pt]
		
		\node[state] 		 (A) at (-6.5,-1) [fill=blue!60]    {$S_1$};
		\node[state]         (B) at (-7.5,-4) [fill=blue!60]    {$S_2$};
		\node[state]         (C) at (-4.5,-2) [fill=green!20] {$S_3$};
		\node[state]         (D) at (-4,-4) [fill=green!20] {$S_4$};
		\node[state]         (E) at (-8,-1) [fill=red!80] {$\good$};
		
		\node[state] 		 (F) at (-.7,-1) [fill=blue!60]    {$S'_1$};
		\node[state]         (G) at (.6,-1) [fill=green!20] {$S'_2$};
		\node[state]         (H) at (-2,-1) [fill=red!80] {$\good$};

		\path (A) edge [bend right=-80,below] node {\scalebox{0.75}{2+$\varepsilon_{13}$}} (C)
		(C) edge [bend right=-38,above,pos=.26] node{\scalebox{0.72}{$\Lambda_{31}+\varepsilon_{31}$}} (A)
		(A) edge [bend right=60] node{\scalebox{0.75}{2+$\varepsilon_{15}$}} (E)
		(B) edge [bend right=20,below,pos=.35] node{\scalebox{0.75}{1+$\varepsilon_{23}$}} (C)
		(B) edge [bend right=-35] node{\scalebox{0.75}{2+$\varepsilon_{25}$}} (E)
		(C) edge [bend right=20,above,pos=.4] node{\scalebox{0.75}{1+$\varepsilon_{32}$}} (B)
		(B) edge [bend right=-0,below] node{\scalebox{0.75}{$1+\varepsilon_{24}$}} (D)
		(D) edge [bend right=-60] node{\scalebox{0.75}{$\Lambda_{42}+\varepsilon_{42}$}} (B);ac
		\path (F) edge [bend right=-65] node {2} (G)
		(F) edge [bend right=0] node{2} (H)
		(G) edge [bend left=65] node{1} (F);
		\node at (-1,-3.5) {$
			P_\bisim = \begin{bmatrix}
			1 & 0 \\
			1 & 0 \\
			0 & 1 \\
			0 & 1
			\end{bmatrix}
			$};
		
		\end{tikzpicture}
		\caption{Full state $\varepsilon$-perturbed CTMC (left), reduced-order CTMC (right), and projection matrix (right, below)
                         computed for the unperturbed CTMC ($\varepsilon_{ij} = 0$) with $\Lambda_{31} = 0$ and $\Lambda_{42}=1$.} 
                         \label{example_full}
	\end{figure}

\subsection{Generalised Projections and Reduction}

Suppose we are given CTMCs $\M$ and $\bar{\M}$, with corresponding dynamical systems \eqref{eq:Xp} and \eqref{eq:Xr}, and a 
matrix $P$ with entries in $[0,1]$ whose rows add up to $1$,
such that $A P = P \bar A$.
We call such a $P$ a \emph{generalised projection}.
Define vector $\bar C_S := C_S P$.
In general, the equality $\B = P\bar{\B}$ does
not hold for generalised projections.
In the following we provide a method based on Lyapunov stability theory to quantify an upper bound $\varepsilon(t)$ such that
\begin{equation}
\label{target}
\left|Prob^{\M}(C_S,t)-Prob^{\bar{\mathcal M}}(\bar C_S,t)\right|\le\varepsilon(t)
\end{equation}
for all $t\geq 0$, where $\varepsilon(t)$ depends linearly on the mismatch $\B -
P\bar{\B}$ and decays exponentially with $t$.

First, we recall some basic results for linear dynamical systems (see, e.g., \cite{Doyle:1990}). The dynamics of these systems are represented by a set of linear differential equations of the form 
\begin{equation}
\label{eq:linear}
\frac{d}{dt}Y(t) = A Y(t),\quad Y(t)\in\mathbb R^m,\quad Y(0) = Y_0.
\end{equation}
We call the system stable if $A$ is a stable matrix. In this case, it is known that $\lim_{t\rightarrow \infty} Y(t) = 0$ for any initial state $Y(0) = Y_0\in\reals^m$.
\begin{definition}
\label{def:Lyap}
A continuous scalar function $V:\mathbb R^{m}\rightarrow\mathbb R$ is called a \emph{Lyapunov function} for 
the dynamical system \eqref{eq:linear} if 
$V(y) = 0$ for $y=0$;
$V(y)>0$ for all $y\in\mathbb R^m\backslash\{0\}$;
and $dV(Y(t))/dt<0$ along trajectories of the dynamical system with $Y(t)\ne 0$.
\end{definition}

A matrix $M\in\mathbb R^{m\times m}$ is symmetric if $M^T = M$.
A symmetric matrix $M$ satisfying the condition $Y^T M Y>0$ 
for all $Y\in\mathbb R^m\backslash\{0\}$ is called
\emph{positive definite}, and written as $M\succ 0$.
Any symmetric matrix $M$ satisfying $Y^T M Y\ge 0$ for all $Y\in\mathbb R^m$ is called \emph{positive semi-definite}, written as $M\succeq 0$.
Similarly, we can define \emph{negative definite} matrices $M\prec 0$ and negative semi-definite matrices $M\preceq 0$.
We write $M_1\succ M_2$ if and only if $M_1-M_2\succ 0$ and $M_1\succeq M_2$ if and only if $M_1-M_2\succeq 0$. $M_1\prec M_2$ and $M_1\preceq M_2$ are defined similarly. The eigenvalues of a symmetric positive definite matrix $M$ are always positive. We denote the largest eigenvalue of the positive definite matrix $M$ by $ \lambda_{max}(M)$. Any positive definite matrix $M$ satisfies $Y^T M Y\le \lambda_{max}(M)\|Y\|_2^2$ for any $Y\in\mathbb R^m$, where $\|Y\|_2$ indicates the two norm of $Y$.
The following is standard.

\begin{theorem}{\cite{Khalil:1996}}
\label{th:stability}
Linear dynamical system \eqref{eq:linear} 
is stable iff
there exists a quadratic Lyapunov function $V(Y)=Y^TMY$ 
such that $M\succ 0$ and $A^T M + MA\prec 0$. 
Moreover, for any constant $\kappa>0$ such that $A^T M + MA + 2\kappa M\preceq 0$, we have
\begin{equation*}
\|Y(t)\|_2 \leq L e^{-\kappa t} \|Y_0\|_2,\quad \forall Y_0\in\mathbb R^m,\forall t\in\mathbb R_{\ge 0},
\end{equation*}
for some constant $L\ge 0$, where $\|\cdot\|_2$ indicates the two norm of a vector.
\end{theorem}

Note that in our setting, we are not interested in the study of asymptotic stability of systems, but we are 
given two dynamical systems \eqref{eq:Xp} and \eqref{eq:Xr}, and we would like to know how close their trajectories 
are as a function of time. 
In this way we can use one of them as an approximation of the other one with guaranteed error bounds.
For this reason, we define Lyapunov function $V:\mathbb R^m\times\mathbb R^r\rightarrow\mathbb R$ of the form
\begin{equation}
\label{eq:Lyapunov}
V(X,\bar X)=(X-P\bar X)^TM(X-P\bar X),
\end{equation}
where $M\succ 0$ is a positive definite matrix.
The value of $V(X(t),\bar X(t))$ at $t=0$ can be calculated as
\begin{align}
V(X(0),\bar X(0))& = (A^{-1}\B-P\bar A^{-1}\bar{\B})^TM(A^{-1}\B-P\bar A^{-1}\bar{\B})\nonumber\\
& = (\B- P\bar{\B})^T{A^{-1}}^T M A^{-1}(\B-P\bar{\B}),\label{eq:V0}
\end{align}
where the second equality is obtained using $AP=P \bar A$ which implies $P {\bar A}^{-1} = A^{-1} P$.
The next theorem shows that the function \eqref{eq:Lyapunov} is indeed a Lyapunov function that satisfies the conditions of Definition~\ref{def:Lyap} but for the dynamical equations of $(X(t)-P\bar X(t))$.

\begin{theorem}
\label{th:BMI}
Consider dynamical systems \eqref{eq:Xp} and \eqref{eq:Xr} with invertible matrix $A$, and let $P$ be
a generalised projection satisfying $AP = P\bar A$.
If there exist matrix $M$ and constant $\kappa>0$ satisfying the following set of matrix inequalities:
	\begin{equation}
	\label{LMI0}
	\left\{
	\begin{array}{lr}
	M\succ 0\\
	C_S^TC_S\preceq M\\
	MA+A^TM+2\kappa M\preceq 0,
	\end{array}\right.
	\end{equation}
then we have
$\| \Prob^\M(C_S, t) - \Prob^{\bar{\M}}(\bar{C}_S, t)\|_2 \leq \varepsilon(t)$, for all $t\geq 0$, with
	\begin{equation}
	\label{eq:error}
	\varepsilon(t) = \xi\|\varGamma\|_2 e^{-\kappa t},
	\end{equation}
	where $\varGamma := \B - P\bar{\B}$ is the mismatch induced by the generalised membership functions and $\xi^2 := \lambda_{max}({A^{-1}}^T M A^{-1})$.
\end{theorem}

	The error in \eqref{eq:error} is exponentially decaying with decay factor $\kappa$ and
	increases linearly with mismatch $\varGamma$. A different version of the result,
	is proved in Appendix~\ref{sec:app1}.
\begin{proof}
With the abuse of notation, let us denote $V(X(t),\bar X(t))$ under the dynamics of $X(t)$ and $\bar X(t)$ in \eqref{eq:Xp} and \eqref{eq:Xr} also by $V(t)$:
\begin{equation*}
V(t):=V(X(t),\bar X(t)),\quad \forall t\ge 0.
\end{equation*}
 We assume the argument of $V$ can be inferred from the context, which is either a time instance $t$ or the pair $(X,\bar X)$.
	We compute derivative of $V(t)$ with respect to time:
	\begin{align*}
	\frac{d}{dt}V(t)& =\frac{d}{dt}V(X(t),\bar X(t)) = \frac{dV(X,\bar X)}{d(X-P\bar X)}\frac{d(X-P\bar X)}{dt}&& \text{(using chain rule for derivatives)}\\
	& \frac{d(X-P\bar X)^T}{dt}M(X-P\bar X) + (X-P\bar X)^T M\frac{d(X-P\bar X)}{dt} && \text{(using definition of $V(X,\bar X)$)}\\
	 & = X^TMAX+X^TA^TMX-X^TMP\bar A\bar X&& \text{(using dynamics \eqref{eq:Xp}-\eqref{eq:Xr})}\\
	&-X^TA^TMP\bar X-\bar X^T\bar A^TP^TMX-\bar X^TP^TMAX\\
	&+\bar X^TP^TMP\bar A\bar X+\bar X^T\bar A^TP^TMP\bar X.
	\end{align*}
	Because of equality $AP=P\bar A$, we can factorise $\frac{d}{dt}V + 2\kappa V$ as
	\begin{equation}
	\label{eq:weighting}
	\frac{d}{dt}V + 2\kappa V=[X^T \bar X^T]
	\begin{bmatrix}
	K_{11} & K_{12}\\
	K_{21} & K_{22}
	\end{bmatrix}
	\begin{bmatrix}
	X\\
	\bar X
	\end{bmatrix},
	\end{equation}
	where
	\begin{align}
	\label{K_mat}
	K_{11} & =MA+A^TM+2\kappa M\\
	K_{12}&=K_{21}^T=-MP\bar A-A^TMP-2\kappa MP\\
	K_{22}&=P^TMP\bar A + \bar A^TP^TMP+2\kappa P^TMP.
	\end{align}
	We can decompose the weight matrix in \eqref{eq:weighting} as
	\begin{align*}
	\begin{bmatrix}
	K_{11} & K_{12}\\
	K_{21} & K_{22}
	\end{bmatrix}
	&=
	\begin{bmatrix}
	K_{11} & -K_{11}P\\
	-P^TK_{11}^T & P^TK_{11}P
	\end{bmatrix}
	=
	\begin{bmatrix}
	\I\\
	-P^T
	\end{bmatrix}
	K_{11}
	\begin{bmatrix}
	\I & -P
	\end{bmatrix}.
	\end{align*}
        Recall from inequalities of \eqref{LMI0} that $K_{11}$ satisfies $K_{11}=MA+A^TM+2\kappa M\preceq 0$, which implies $\frac{d}{dt}V+2\kappa V\leq0.$
	This inequality guarantees that $V(t)\le V(0)e^{-2\kappa t}$. Note that since $V(t) = V(X(t),\bar X(t))$ is a quadratic function of $X(t)-P\bar X(t)$, the inequality $V(t)\le V(0)e^{-2\kappa t}$ means $X(t)-P\bar X(t)$ will go to zero exponentially in time with decaying factor $\kappa$. To get a precise upper bound on error between outputs of the two systems, we first bound $V(0)$.
	Notice that $V(0)$ is obtained in \eqref{eq:V0}, which satisfies
	\begin{equation*}
	V(0) = V(X(0),\bar X(0)) =  \varGamma^T({A^{-1}}^T M A^{-1})\varGamma \le \lambda_{max}({A^{-1}}^T M A^{-1})\|\varGamma\|_2^2.
	\end{equation*}
	The inequality holds since $M$ is positive definite which makes ${A^{-1}}^T M A^{-1}$ also positive definite.
	Now recall $\bar C_S := C_S P$ and write
	\begin{align*}
	\| \Prob^\M(C_S, t) & - \Prob^{\bar{\M}}(\bar{C}_S, t)\|_2 
	 = \|C_S X(t)-\bar C_S X(t)\|_2 \\
	& =  \|C_S (X(t)-P\bar X(t))\|_2 && \text{(using $\bar C_S := C_S P$)}\\
	& =\left[(X(t)-P\bar X(t))^T{C_S}^TC_S(X(t)-P\bar X(t))\right]^{1/2}&& \text{(using equality $\|Y\|_2 = [Y^TY]^{1/2}$)}\\
	& \le  \left[(X(t)-P\bar X(t))^T M (X(t)-P\bar X(t))\right]^{1/2} \quad && \text{(using $C_S^TC_S\preceq M$)}\\
	 & = V(t)^{1/2} \quad && \text{(using definition of $V(t)$}\\
	 & \le V(0) ^{1/2}e^{-\kappa t} \quad && \text{(using the exponential bound on $V(t)$)}\\
	 & \le \lambda_{max}({A^{-1}}^T M A^{-1})^{1/2}\|\varGamma\|_2 e^{-\kappa t} && \text{(using the bound on $V(0)$)}\\
	 &  = \xi\|\varGamma\|_2 e^{-\kappa t} = \varepsilon(t) && \text{(using definitions of $\xi$ and $\varepsilon(t)$)}.
	\end{align*}
	This completes the proof.
\end{proof}

Matrix inequalities \eqref{LMI0} in Theorem~\ref{th:BMI} are bilinear in terms of unknowns (entries of $M$ and constant $\kappa$)
due to the multiplication between $\kappa$ and $M$, thus are difficult to solve. Under Assumption~\ref{ass:invert}, there exists $M$ and $\kappa$ such that \eqref{LMI0} is satisfied.
In the following we show how to obtain a solution efficiently when $A$ is stable. 
\begin{theorem}
		\label{th:algo-reduction1}
		Assumption~\ref{ass:invert} implies that matrix $A$ has a simple eigenvalue equal to $\bar\rho := \max_i\realpart(\lambda_i(A))$ and its left eigenvector $\nu$ can be selected such that all its entries are strictly positive.
		A feasible solution of \eqref{LMI0} can be selected by letting $\kappa$ be any positive constant 
		\begin{equation}
		\label{kappa2}
		\kappa\le -\frac{1}{2}\bar\rho = -\frac{1}{2}\max_i\realpart(\lambda_i(A)),
		\end{equation}
		and choosing the diagonal matrix $M = diag(\nu)$ with entries of $\nu$ normalised to have them greater or equal to one. 
	\end{theorem}
	\begin{proof}
	The matrix $H=\frac{A}{\gamma}+\mathbb I_m$ is sub-stochastic and irreducible. According to Perron-Frobenius theorem, $H$ has a simple real eigenvalue $\rho$, which is its largest eigenvalue in absolute sense, and its associated left eigenvector $\nu$ having strictly positive entries. Without loss of generality, we assume that $\nu$ is normalised such that it has all entries greater or equal to one. We also proved in Proposition~\ref{prop:stable} that $\rho<1$. The definition of $H$ implies that $\lambda_i(H) = \lambda_i(A)/\gamma + 1$. Thus we get $\bar\rho := \max_i\realpart(\lambda_i(A)) = -\gamma(1-\rho)$ is a simple eigenvalue of $A$ with the same left eigenvector $\nu$.\\
	 Matrix $(A+\kappa\mathbb I_m)$ has exactly the same eigenvalues as that of $A$ but increased by $\kappa$.
	Selecting $\kappa<-\bar\rho$ implies $(A+\kappa\mathbb I_m)$  still has all its eigenvalues with negative real parts. Therefore, $(A+\kappa\mathbb I_m)$ is stable and according to \autoref{th:stability}, there is a matrix $M$ satisfying $(A^T+\kappa \mathbb I_m)M + M(A+\kappa\mathbb I_m) \prec 0$, which means $A^TM + MA+2\kappa M \preceq 0$.\\
	We show that $M = diag(\nu)$ is a solution for this inequality. Denote by $\mathbf 1_m$ the column vector of dimension $m$ with all entries equal to one. We have
	\begin{align*}
	&(A^T+2\kappa \mathbb I_m)M\mathbf 1_m = (A^T+2\kappa \mathbb I_m)\nu^T = (\bar\rho + 2\kappa)\nu^T,\\
	& MA\mathbf 1_m = M(A\mathbf 1_m) = (\nu^T)\cdot (A\mathbf 1_m)\quad  \text{(entry-wise product of $\nu^T$ and $A\mathbf 1_m$).}
	\end{align*}
	Since $\nu^T$ has positive entries, $(\bar\rho + 2\kappa)\le 0$, and $(A\mathbf 1_m)$ has non-positive entries, we have that both matrices $(A^T+2\kappa \mathbb I_m)M$ and $MA$ are right sub-stochastic satisfying Assumption~\ref{ass:invert}. Therefore, $(A^T+2\kappa \mathbb I_m)M + MA$ is symmetric and stable, its eigenvalues will be negative, thus it is semi-definite negative. This concludes the proof.    
	\end{proof}
	
Next, we show that for a given $r\leq m$, we can find a suitable $\bar{A}$ and $P$ such that $AP=P\bar{A}$ .

\begin{theorem}
\label{th:algo-reduction2}
Given the matrix $A\in \mathbb R^{m\times m}$, for each $r\leq m$, there is a $m\times r$ matrix $P$ and an
$r\times r$ matrix $\bar{A}$, computable in polynomial
time in $m$, such that $AP = P\bar{A}$.
\end{theorem}
\begin{proof}
Every matrix $A$ can be decomposed as
\begin{equation}\label{Schur}
A=UNU^{-1},
\end{equation}
in which $N$ is an upper triangular matrix, called the \emph{Schur form} of $A$, and $U$ is a unitary matrix \cite{Horn:1985}. 
Schur decomposition of $A$ can be performed
iteratively with $\mathcal{O}(m^3)$ 
arithmetic operations using QR decomposition \cite{Dem:1997}. 
We choose $\bar{A}$ as the first $r$ rows and columns
of $N$ and $P$ as first $r$ 
columns of $U$.
Since $N$ is upper triangular,
the equality $AP=P\bar{A}$ holds for this choice of $\bar{A}$ and $P$.
\end{proof}

Once $\kappa$ is fixed, constraints \eqref{LMI0} become
matrix inequalities that are linear in terms of entries of $M$  and can be solved using convex optimisation
\cite{Feller:1968} and developed tools for linear matrix inequalities \cite{Grant:2008,Lofberg:2004}.
In particular, the diagonal matrix $M$ defined in \autoref{th:algo-reduction1} is a feasible solution
to the matrix inequalities.
However, when $C_S$ is not full rank, which is the case when $S_0\ne S$, solving the matrix inequalities for $M$ can result in better error bounds.

Notice that $V(0)=(X(0)-P\bar{X}(0))^TM(X(0)-P\bar{X}(0))$ and using 
\eqref{eq:Xp}, we have $X(0)=A^{-1}\B$. 
Therefore, it is important to find $\bar{X}(0)$ that results in the least $V(0)$. 
We can compute $\bar{X}(0)$ by minimising $V(0)$:
\begin{equation}
\label{eq:X0_opt}
\min_{\bar{X}(0)}\,\, \left[X(0)-P\bar{X}(0)\right]^TM\left[X(0)-P\bar{X}(0)\right],
\end{equation}
which is a weighted least square optimisation and has the closed-form solution
\begin{equation}
\label{eq:X0_1}
\bar{X}(0)=(P^TMP)^{-1}P^TM(A^{-1}\B).
\end{equation}
Choosing this initial state $\bar{X}(0)$ will provide a tighter initial error bound.
Knowing $\bar{A}$ and $\bar{X}(0)$, one can find $\bar{\B}=\bar{A}\bar{X}(0)$.

Theorems~\ref{th:algo-reduction1}-\ref{th:algo-reduction2} give an algorithm, shown in Algorithm~\ref{alg:CTMC_red}, to find lower
dimensional approximations to the dynamical system \eqref{eq:Xp},
and Theorem~\ref{th:BMI} provides a quantitative error bound for the approximation.
The procedure is summarised in Algorithm~\ref{alg:CTMC_red}.
Given a time-bounded reachability problem and an error bound $\varepsilon$,
we iteratively compute reduced order dynamical systems of dimension $r = 1, \ldots, m-1$
using Theorems~\ref{th:algo-reduction1}-\ref{th:algo-reduction2}.
Then, we check if the error bound in Theorem~\ref{th:BMI} is at most $\varepsilon$.
If so, we solve the dynamical system of dimension $r$ (using, e.g., exponential of an upper-triangular matrix)
to compute an $\varepsilon$-approximation to the time bounded reachability problem.
If not, we increase $r$ and search again. 
\begin{example}
		Consider the CTMC in Fig.~\ref{example_full} with  $\Lambda_{31}=1$, $\Lambda_{42} = 2$ and $\varepsilon_{ij} = 0$ for all $i,j$ (the CTMC is unperturbed). The generator matrix for the CTMC is
		\begin{equation*}
		Q=
		\begin{bmatrix}
		-4 & 0 & 2 & 0 & 2\\
		0 & -4 & 1 & 1 & 2\\
		1 & 1 & -2 & 0 & 0\\
		0 & 2 & 0 & -2 & 0\\
		0 & 0 & 0 & 0 & 0
		\end{bmatrix}.
		\end{equation*}
		As in Example~\ref{example_full}, we are interested in computing the probability of
		reaching the state $\good$. Using the partition defined in Eq.~\eqref{eq:Q_part}, we get
		\begin{equation*}
		A=
		\begin{bmatrix}
		-4 & 0 & 2 & 0\\
		0 & -4 & 1 & 1\\
		1 & 1 & -2 & 0\\
		0 & 2 & 0 & -2
		\end{bmatrix},\quad \boldsymbol{\beta}=
		\begin{bmatrix}
		2\\
		2\\
		0\\
		0
		\end{bmatrix}.
		\end{equation*}
		Note that $A$ is reducible with $\bar\rho = -0.7639$. All the values are reported by rounding to 4 decimal digits. We select the decay rate $\kappa = 0.3820$ using Eq.~\eqref{kappa2}.
		Then we compute $U$ and $N$ based on the Schur decomposition of $A$:
		\begin{equation*}
		N=
		\begin{bmatrix}
		-5.2361  & 0  &  0.1602 &  -0.9871\\
		0  & -0.7639  & -0.9871  & -0.1602\\
		0    &     0  & -4.4142  &  0\\
		0     &    0    &     0  & -1.5858
		\end{bmatrix},\quad U=
		\begin{bmatrix}
		0.6015  &  0.3717  &  0.6533  & -0.2706\\
		0.6015   & 0.3717  & -0.6533   & 0.2706\\
		-0.3717  &  0.6015  & -0.2706  & -0.6533\\
		-0.3717   & 0.6015 &   0.2706  &  0.6533
		\end{bmatrix}.
		\end{equation*}
		Using \autoref{th:algo-reduction1} we find matrix $M$ as
		\begin{equation*}
		M=\begin{bmatrix}
		1 &0 &0 &0\\
		0& 2 &0 & 0\\
		0&0&3.2361&0\\
		0&0&0&1.6180
		\end{bmatrix}.
		\end{equation*}
		Selecting the order $r=2$, we find $\bar A$ as the first $(2\times 2)$ block of $N$ and $P$ the first $2$ columns of $U$:
		\begin{equation*}
		\bar A=
		\begin{bmatrix}
		-5.2361  & 0\\
		0  & -0.7639 
		\end{bmatrix},\quad P=
		\begin{bmatrix}
		0.6015  &  0.3717\\
		0.6015   & 0.3717 \\
		-0.3717  &  0.6015\\
		-0.3717   & 0.6015 
		\end{bmatrix}.
		\end{equation*}
		Using \eqref{eq:X0_1}, we compute the initial state of the reduced-order system as
		\begin{equation*}\bar X(0)=
		\begin{bmatrix}
		0.4595\\
		1.9465
		\end{bmatrix}.
		\end{equation*}
		The above selection results in $\varepsilon(T)=0$ for any arbitrary time bound $T$. Therefore, the order of the set of differential equations that we need to solve reduces from four into two without incurring any error. In this case, our approach retrieves the reduction originating from the exact bisimulation.\\
		 We now consider a perturbed version of  the CTMC with the generator matrix
		\begin{equation}
		\label{eq:Q_perturbed2}
		Q=
		\begin{bmatrix}
		-3.95 & 0 & 1.95 & 0 & 2\\
		0 & -4.05 & 1.05 & 1 & 2\\
		1 & 1 & -2 & 0 & 0\\
		0 & 2 & 0 & -2 & 0\\
		0 & 0 & 0 & 0 & 0
		\end{bmatrix}.
		\end{equation}
		By performing the same computations as above, we find
		\begin{equation*}
		\kappa=0.3730\qquad
		M=\begin{bmatrix}
		1 &0 &0 &0\\
		0& 1.9047 &0 & 0\\
		0&0&3.1887&0\\
		0&0&0&1.5376
		\end{bmatrix},
		\end{equation*}
		and
		\begin{equation*}
		\bar A=
		\begin{bmatrix}
		-5.2580  & -0.0770\\
		0  & -0.7613 
		\end{bmatrix},\quad P=
		\begin{bmatrix}
		0.5436  &  0.3753\\
		0.6443   & 0.3864\\
		-0.3646  &  0.5922\\
		-0.3955   & 0.5993
		\end{bmatrix}, \quad\bar X(0)=
		\begin{bmatrix}
		0.4165\\
		1.9454
		\end{bmatrix}.
		\end{equation*}
		For example, we have $\varepsilon(T)=0.0008e^{-0.3730T}$ according to Theorem~\ref{th:BMI}, which is $0.0005$ for time bound $T=1$.
		\label{example_ctmc_red_Lyap}
	\end{example}

\begin{algorithm}	
	\caption{Order reduction of CTMCs}
	\label{alg:CTMC_red}
	\SetAlgoLined
	\KwIn{CTMC $\M = (S_{\mathcal{M}},R,\alpha)$, time bound $T$, maximum error bound $\varepsilon$}
	\begin{enumerate}
		\item Compute $A$, $\boldsymbol{\beta}$ and $\kappa$, based on \eqref{eq:Q_part} and \eqref{kappa2}
		\item Compute $M$ using Theorem~\ref{th:algo-reduction1}
		\item Compute the Schur decomposition of $A$ and save the matrices $U$ and $N$ using \eqref{Schur}\\
		\item $r\leftarrow 0$\\
		\item \textbf{Do}\\
		\quad $r\leftarrow r+1$\\
		\quad Set $\bar{A}$ as the first $r$ rows and columns
		of $N$\\
		\quad Set $P$ as first $r$ 
		columns of $U$\\
		\quad Compute $\bar X(0)$ according to \eqref{eq:X0_1}\\
		\quad Compute error bound $\varepsilon_r$ using \eqref{eq:error} for time bound $T$ and $\bar{\B}=\bar{A}\bar{X}(0)$\\
		\textbf{While} ($\varepsilon_r > \varepsilon$)
	\end{enumerate}
	\KwOut{Reduced-order system \eqref{eq:Xr}}	
\end{algorithm}

\subsection{Symbolic Computation on the Reduced Model}
\label{sec:app0}
Based on the construction of $\bar A$ of the reduced system according to the Schur form \eqref{Schur}, matrix $\bar A$ is upper-triangular as
\begin{equation*}
\bar A=
\begin{bmatrix}
\bar A_{11} & \bar A_{12} & \bar A_{13}&\cdots &\bar A_{(1)(r-1)}& \bar A_{1r}\\
0 & \bar A_{22} &\bar A_{23}&\cdots &\bar A_{(2)(r-1)}& \bar A_{2r}\\
0 & 0 &\bar A_{33} &\cdots &\bar A_{(3)(r-1)}& \bar A_{3r}\\
\vdots & \vdots &\vdots&\ddots &\vdots&\vdots\\
0 & 0 & \cdots &\cdots& \bar A_{(r-1)(r-1)}&\bar A_{(r-1)(r)}\\
0 & 0 &\cdots& \cdots &0& \bar A_{rr}
\end{bmatrix}.
\end{equation*}
This property of $\bar A$ can be exploited to make the computation of reachability probability more efficient.
 In fact, solution of the differential equation $\dot{\bar X}(t)=\bar A \bar X(t)$ in \eqref{eq:Xr} can be written as $\bar X(t) = e^{\bar At}\bar X(0)$. Let us first assume all diagonal elements of $\bar A$ are distinct. Denote the $i^{\text{th}}$ element of $\bar X(t)$ by $\bar X_i(t)$. The last element of $\bar X(t)$ can be easily computed as
 \begin{equation*}
	\dot {\bar X}_r(t)=\bar A_{rr}\bar X_r(t)\Rightarrow\bar X_r(t)=e^{\bar A_{rr}t}\bar X_r(0).
\end{equation*}
In general, it is possible to perform the computations bottom-up. Once we solve the equations for ${\bar X}_r(t),{\bar X}_{r-1}(t),\ldots, {\bar X}_{i+1}(t)$, we use their explicit form to solve the differential equation for $\bar X_i(t)$. This gives the solution in closed-form as
\begin{equation}\label{eq:sum_exp}
	\bar X_i(t)=\sum_{j=i}^r \alpha_{ij}e^{\bar A_{jj}t},
\end{equation}
where
\begin{equation*}
\alpha_{ij}=\begin{cases}
\sum_{k=i+1}^{j}\frac{\bar A_{ik}\alpha_{kj}}{-\bar A_{ii}+\bar A_{jj}} & \text{for }\,\, j>i,\\
-\sum_{j=i+1}^{r}\alpha_{ij}+\bar X_i(0) & \text{for }\,\, j=i.
\end{cases}
\end{equation*}
This closed-form solution can be verified inductively. Note that the computation of $\alpha_{ij}$ is performed sequentially and backward with respect to the index $i$. To make these computations clear, let us define the matrix $\boldsymbol{\alpha} := [\alpha_{ij}]_{i,j}$, which is upper triangular. The last row of this matrix has one non-zero element, which is simply $\alpha_{rr} = \bar X_r(0)$. The computation of the $i^{th}$ row of $\boldsymbol{\alpha}$ is performed as follows. The non-diagonal elements in the $i^{th}$ row will need the entries from previously computed rows which are the $(i+1)^{st},(i+2)^{nd},\ldots,r^{th}$ rows. The diagonal element in the $i^{th}$ row needs its non-diagonal elements.

For the case that $\bar A$ has eigenvalues with multiplicities $\mathsf{m}>1$, the closed-form solution \eqref{eq:sum_exp} becomes a linear combination of functions $t^le^{\bar A_{ii}t}$ for $0\leq l\leq \mathsf{m}-1$, and the coefficients can be computed in a similar way. The details of such computations can be found in general text books on control theory, e.g., \cite{Ogata:2001}.
\begin{example}
	Let us consider the CTMC with the generator matrix given in \eqref{eq:Q_perturbed2}. The reduced system for this CTMC was computed in Example~\ref{example_ctmc_red_Lyap}. We use our symbolic computation method described above to find the solution to the time bounded reachability problem. Based on Eq.~\eqref{eq:sum_exp}, the closed form solution to the time bounded reachability problem over the reduced system with time bound $T$ will be
	\begin{align*}
		\bar X_1(T)&=-0.0332e^{-0.7613T}+0.4498e^{-5.2580T}\\
		\bar X_2(T)&=1.9454e^{-0.7613T}.
	\end{align*} 
	\label{example_ctmc_symbolic}
\end{example}

\section{Time-Bounded Reachability on CTMDPs}
\label{sec:CTMDPs}

First, we define continuous-time Markov decision processes (CTMDPs), which include non-deterministic choice of actions on top of probabilistic jumps.
We use \emph{decision vectors} in the definition of CTMDPs, which are vectors of actions selected in different states. This definition is more suitable for our analysis in this section.

\begin{definition}
	A continuous-time Markov decision process (CTMDP) $\N = (S_{\N},\mathcal D,R_d)$ consists of 
	a finite set $S_{\N}=\{1,2,\ldots,\mathfrak n\}$ of states for some positive natural number $\mathfrak n$,
	a finite set of possible actions $\mathcal D$,
	and action-dependent rate matrices $R_d$, where $d\in \mathcal D^{|S_{\N}|}$ is a \emph{decision vector} containing actions taken at different states, $d := \{d(s)\,|\,s\in S_{\N}\}$.
\end{definition}
Note that some of the actions may not be available at all states. Denote the set of possible decision vectors by $\mathbf D\subseteq \mathcal D^{|S_{\N}|}$.
Similar to CTMCs, we assign an initial distribution $\alpha$ to the CTMDP $\N$. For any fixed $d\in\mathbf D$, $\N_d = (S_{\N},R_d,\alpha)$ forms a CTMC, for which we can define infinitesimal generator  $\bar{\mathbf Q}_d := R_d - diag_s(E_d(s))$ with total exit rates $E_d(s)$ at state $s$
defined as $E_d(s) :=\sum_{s'\in S_{\mathcal{N}}}{R_d(s,s')}$.

A path $\omega$ of a CTMDP $\N$ is a (possibly infinite) sequence including transitions of the form $s_{i}$ $\xrightarrow[\text{}]{\text{$d_i$,$t_i$}}$ $s_{i+1}$, for $i=0,1,2,\ldots$, where $t_i \in \mathbb{R}_{\ge 0}$ is the sojourn time in $s_i$ and $d_i \in \mathcal D$ is a possible action taken at $s_i$. We denote the set of all finite paths of CTMDP $\N$ by $Paths(\N)$. 
 A \emph{policy} provides a mapping from $Paths(\N)\times R_{\geq0}$ to actions of the model, in order to resolve the nondeterminism that occurs in the states of a CTMDP for which more than one action is possible. 
 \begin{remark}
 	We have considered the class of timed positional deterministic policies which suffices for maximising the time-bounded reachability probability \cite{Rabe:2011}. A policy in this class gives the action as a function of the current state and the total passed time.
 \end{remark}

Let $\N = (S\uplus\set{\good,\bad}, \mathcal D, R_d)$ be a CTMDP
with two absorbing states $\good$ and $\bad$, where $|S|=m$, and let $T \in\reals_{\geq 0}$ be a time bound and $\theta\in(0,1)$ a probability threshold.
We are interested in synthesising a policy $\pi$ such that probability of reaching state $\good$ while avoiding state $\bad$ within time interval $[0,T]$ is at least $\theta$ for the CTMDP with initial state $s$:
\begin{equation}
\label{eq:reach_CTMDP}
\Prob^{\N(\pi)}(\mathbf 1(s), T) = \bar{\pi}^{\N(\pi)}_s(T)(\good)\ge \theta,
\end{equation}
where $\text{Prob}^{\N(\pi)}$ is the probability measure induced on paths of $\N$ by resolving non-determinism via policy $\pi$.
Synthesising such a policy can be done by maximising the left-hand side of \eqref{eq:reach_CTMDP} on the set of policies and then comparing the optimal value with $\theta$.
Characterisation of the optimal policy is performed as follows \cite{Buch:2011}.
We partition any generator matrix $Q_d$ corresponding to decision vector $d\in \mathbf D$, as
\begin{equation}
\label{eq:partition}
Q_d=
\begin{bmatrix}
A_d & \vdots & \boldsymbol{\chi}_d &\vdots & \boldsymbol{\beta}_d\\
\dots & \dots & \dots &\dots & \dots\\
\mathbf{0} &\vdots & \mathbf 0 &\vdots & \mathbf 0
\end{bmatrix}
\end{equation}
with $A_d\in\mathbb R^{m\times m}$, $\boldsymbol{\chi}_d\in\mathbb R^{m\times 1}$, and $\boldsymbol{\beta}_d\in\mathbb R^{m \times1}$.
Then for a CTMDP $\N$ with matrix $C$ indicating a subset of initial states $S_0\subseteq S$ for which we would like to satisfy \eqref{eq:reach_CTMDP},
$\max_\pi \Prob^{\N(\pi)}(C, T)$ can be characterised backward in time as the solution of the following set of nonlinear differential equations
\begin{align}
& \frac{d}{dt} W(t) = \max_{d(t)\in\mathbf D}\;Q_{d(t)} W(t),\quad  W(0) = \mathbf 1(\good),\nonumber\\
&\max_\pi \Prob^{\N(\pi)}(C, T) = C W(t),
\label{eq:diff_CTMDP1}
\end{align}
where $W(t)$ is a column vector containing probabilities $\max_\pi \Prob^{\N(\pi)}(\mathbf 1(s), T)$ as a function of initial state $s$.

With respect to the partitioning \eqref{eq:partition}, it is obvious that in \eqref{eq:diff_CTMDP1},  $W(t)(\bad) = 0$ and $W(t)(\good) = 1$ for all $t\in \reals_{\geq 0}$. The remaining state variables $W_{S}(t)$ should satisfy
\begin{align}
& \frac{d}{dt} W_{S}(t) = \max_{d(t)\in \mathbf D}(A_{d(t)} W_{S}(t) + \boldsymbol{\beta}_{d(t)}),\quad  W_{S}(0) = 0,\nonumber\\
& \max_{\pi} Prob^{\N(\pi)}(C_S,t) = C_S W_{S}(t).
\label{eq:diff_CTMDP2}
\end{align}

The optimal policy is the one maximising the right-hand side of differential equation in \eqref{eq:diff_CTMDP2},
$$\pi^\ast = \{d(t)\in\mathbf D\,|\, t\in\mathbb R_{\ge 0}\},$$
thus it is time-dependent and is only a function of state of the CTMDP at time $t$.
In \cite{Rabe:2011}, it is shown that the policy that maximises time-bounded reachability probability of CTMDPs contains only finitely many switches. However, finding the optimal policy is computationally expensive for CTMDPs with large number of states. The current state of the art solutions are based on breaking the time interval $[0,T]$ into smaller intervals of length $\delta$,
and then computing (approximate) optimal decisions in each interval of length $\frac{T}{\delta}$ sequentially (see \cite{Fearnley:2016,Hasan:2015}).
Thus, a set of linear differential equations
must be solved in each interval, which is computationally expensive.

In the following, we will develop a new way of synthesising a policy that satisfies \eqref{eq:reach_CTMDP} by approximating the solution of \eqref{eq:diff_CTMDP2} via generalised projections and reductions. We treat \eqref{eq:diff_CTMDP2} as a \emph{switched affine system} \cite{ANTO:2010}. We are given a collection of $|\mathbf D|$ affine dynamical systems, characterised by the pairs $(A_{d},\boldsymbol{\beta}_{d})$, and the role of any policy $\pi = \{d(t)\in\mathbf D,\,t\ge 0\}$ is to switch from one dynamical system to another by picking a different pair. The main underlying idea of our approximate computation is to consider the reduced order version of these dynamical systems and find a switching policy $\pi$. We provide guarantees on the closeness to the exact reachability probability when this policy is applied to the original CTMDP. For this we require the following assumption.
\begin{assumption}
\label{ass:stab_MDP}
Matrices $\{A_d,\,d\in\mathbf D\}$ are all stable.
\end{assumption}
Note that this assumption is satisfied if for each choice of actions, the resulting CTMC is irreducible (Prop.~\ref{prop:stable})
and the time-bounded reachability problem does not have a trivial solution.

Under Assumption \ref{ass:stab_MDP}, we can find matrix $M_d$ and constant $\kappa_d>0$, for any $d\in\mathbf D$, such that the following matrix inequalities hold:
\begin{equation}
	\label{LMI_MDP}
	\left\{
	\begin{array}{lr}
	M_d\succ 0\\
	C_S^TC_S\preceq M_d\\
	M_dA_d+A_d^TM_d+2\kappa_d M_d\preceq 0,
	\end{array}\right.
	\end{equation}
We need the following lemma that gives us a bound on the solution of reduced order systems.
\begin{lemma}
\label{lem:bound}
Suppose generalised projections $P_d$ and matrices $\bar A_d$ satisfy $A_d P_d = P_d \bar A_d$ for any $d\in\mathbf D$. 
Then $V(\bar X_d) = \bar X_d^T \bar M_d \bar X_d$ with $\bar M_d = P_d^T M_d P_d$ and $M_d$ satisfying \eqref{LMI_MDP}, is a Lyapunov function for $d\bar X_d(t)/dt = \bar A_d\bar X_d(t)$ for each $d\in\mathbf D$. 
Moreover,
\begin{equation}
\label{eq:bounded_X}
\|\bar X_d(t_1)\|_{\bar M_d}\le \|\bar X_d(t_0)\|_{\bar M_d} e^{-\kappa_d (t_1-t_0)},\quad \forall t_1\ge t_0,
\end{equation}
where $\|Y\|_{G}:=\sqrt{Y^TGY}$ is the weighted two-norm of a vector $Y$.
\end{lemma}
\begin{proof}
We prove \eqref{eq:bounded_X} via a bound on the Lyapunov function $V(\bar X_d)$:
\begin{align*}
\frac{d}{dt}V(\bar X_d) & = (\bar A_d\bar X_d)^T\bar M_d\bar X_d  + \bar X_d^T\bar M_d(\bar A_d\bar X_d) && \text{(by replacing derivative of $\bar X_d$ with $\bar A_d\bar X_d$)}\\
& = \bar X_d^T (\bar A_d^T\bar M_d+\bar M_d\bar A_d) \bar X_d && \text{(by factorization)}\\
&= \bar X_d^T(\bar A_d^TP_d^T M_d P_d + P_d^T M_d P_d\bar A_d)\bar X_d&& \text{(by using identity $\bar M_d = P_d^T M_d P_d$)}\\
  &= \bar X_d^T(P_d^T A_d^T M_d P_d + P_d^T M_d A_d P_d)\bar X_d&& \text{(by using $A_d P_d = P_d \bar A_d$)}\\
  & = \bar X_d^TP_d^T (A_d^T M_d + M_d A_d)P_d\bar X_d&& \text{(by factorization)}\\
 &  \le -2\kappa_d \bar X_d^TP_d^T M_d P_d\bar X_d = -2\kappa_d V(\bar X_d)&& \text{(by using inequality~\eqref{LMI_MDP})},
\end{align*}
thus $V(\bar X_d(t))\le V(\bar X_d(t_0))e^{-2\kappa_d (t-t_0)}$, for all $t\ge t_0$, which gives \eqref{eq:bounded_X}.
\end{proof}
Consider an arbitrary time-dependent Markov policy $\pi = \{d(t)\in\mathbf D,\,t\ge 0\}$. Then there is a sequence of decision vectors $(d_0,d_1,d_2,\ldots)$ with switching times $(t_0, t_1, t_2,\ldots)$ such that actions in $d_i$ are selected over time interval $[t_{i-1},t_i)$ depending on the state of $\N$, for any $i=0,1,2,\ldots$ with $t_{-1} = 0$.
We first study time-bounded reachability for $\N$ under policy $\pi$, which can be characterised as the switched system:
\begin{equation}
\frac{d}{dt} W_{S}(t) = A_{d_i} W_{S}(t) + \boldsymbol{\beta}_{d_i},\,\forall t\in[t_{i-1},t_i),\, i=0,1,\ldots
\label{eq:switch1}
\end{equation}
Similar to our discussion on CTMC, we prefer to move constant inputs $\boldsymbol{\beta}_{d_i}$ in \eqref{eq:switch1} into initial states. Therefore, we define the following piecewise translation
\begin{equation}
\label{eq:translate}
X(t) := W_{S}(t) + A_{d_i}^{-1}\boldsymbol{\beta}_{d_i},\,\forall t\in[t_{i-1},t_i),\, i=0,1,2,\ldots
\end{equation}
 that depends also on $\pi$.
 Note that $A_{d_i}^{-1}\boldsymbol{\beta}_{d_i}$ is exactly the solution of the unbounded reachability probability (steady state solution of \eqref{eq:switch1} when matrix $A_{d_i}$ is selected for all time instances). Thus the evolution of $X(t)$ becomes
 \begin{equation}
\frac{d}{dt} X(t) = A_{d_i} X(t),\,\,\forall t\in[t_{i-1},t_i),\,\, i=0,1,2,\ldots,
\label{eq:switch2}
\end{equation}
with state $X(t)$ having jumps at switching time instances $t_i$ that are equal to
\begin{equation}
\label{eq:DeltaX}
\Delta X(t_i) := X(t_i) - X(t_i^-) =  A_{d_{i+1}}^{-1}\boldsymbol{\beta}_{d_{i+1}} - A_{d_i}^{-1}\boldsymbol{\beta}_{d_i},
\end{equation}
where $X(t_i^-)$ denotes the left-sided limit of $X(t)$ at $t_i$, i.e., $X(t_i^-): = \lim_{t \uparrow t_i} X(t)$. The quantity $\Delta X(t_i)$ is exactly the difference between unbounded reachability probability if one of the decision vectors $d_i$ and $d_{i+1}$ is taken independent of time.
Similarly, we define 
\begin{equation}
\label{eq:Delta}
\Delta_{ij} := A_{d_j}^{-1}\boldsymbol{\beta}_{d_j} - A_{d_i}^{-1}\boldsymbol{\beta}_{d_i},
\end{equation}
which will be used later in Theorem~\ref{thm:bound2}.
Note that $W_S(t)$ is a continuous function of time no matter what decision vectors $\{d_0,d_1,\ldots\}$ are selected, but it converges to different steady state vectors depending on the chosen decision vectors. On the other hand, when we change the variables to $X(t)$ using the affine transformation \eqref{eq:translate}, $X(t)$ becomes a discontinuous function of time, with discontinuity at time instances $t_i$ and jumps equal to $\Delta X(t_i)$ defined in \eqref{eq:DeltaX}, but it will always converge to zero independent of the chosen decision vectors $\{d_0,d_1,\ldots\}$.

Now we construct the reduced order switched system
\begin{equation}
\frac{d}{dt} \bar X(t) = \bar A_{d_i} \bar X(t),\,\,\forall t\in[t_{i-1},t_i),\,\, i=0,1,2,\ldots,
\label{eq:switch_red}
\end{equation}
with $\bar A_d$ satisfying $A_d P_d = P_d \bar A_d$ for all $d\in\mathbf D$. We choose the values of jumps $\Delta \bar X(t_i) := \bar X(t_i) - \bar X(t_i^-) $ so that the behaviour of \eqref{eq:switch_red} is as close as possible to \eqref{eq:switch2}. For this, we have
\begin{equation}\label{eq:minmax}
\bar X(t_i) := \arg\min_{\bar X} \left\|\Delta X(t_i) - P_{d_{i+1}}\bar X+P_{d_i}\bar X(t_i^-)\right\|_{M_{d_{i+1}}},
\end{equation}
which can be computed for any value of $\bar X(t_i^-)$.

Define the \emph{dwell time} of a policy $\pi$ by $\tau = min_i(t_i-t_{i-1})$, i.e., the minimum time between two consecutive switches of decision vectors in $\pi$.
The paper \cite{NZ10} shows that for any epsilon-optimal policy there is a bound on the minimum dwell time.
The next theorem quantifies the error between the two switched systems using the dwell time of the policy.

\begin{theorem}
	\label{thm:bound2}
	Given a CTMDP $\N$, a policy $\pi$ with dwell time $\tau$, switching time instances $t_0=0\le t_1\le t_2\le \cdots$, and bounded-time reachability over $[0,T]$.
	Suppose there exist $M_{d_i},\kappa_{d_i}$ satisfying \eqref{LMI_MDP}, constant $\mu$ satisfying $M_{d_i}\preceq \mu M_{d_j}$ for all $d_i,d_j\in\mathbf D$, 
	and matrices $\bar A_{d_i}, P_{d_i}$ such that $A_{d_i} P_{d_i} = P_{d_i}\bar A_{d_i}$. 
	Then we have
	\begin{equation}\label{eq:bound_MDP2}
	\|X(T) - P_{d_{n+1}}\bar X(T)\|_{M_{d_{n+1}}}\le\varepsilon_n e^{-\kappa (T-t_n)},
	\end{equation}
	where $t_n$ is the last switching time instance before the time bound $T$ and $\kappa := \min_d\kappa_d$ is the minimum decay rate.
	The quantity $\varepsilon_n$ is obtained from the difference equations
	\begin{align}
	& \bar{\varepsilon}_i=\mu g\bar{\varepsilon}_{i-1}+\Delta_{max}\nonumber\\
	& \varepsilon_i=\mu g\varepsilon_{i-1}+2\mu g\bar{\varepsilon}_{i-1}+2\Delta_{max}, \quad i\in\{1,2,\ldots\},
	\label{eq:difference_equations}
	\end{align}
	where $g:=e^{-\kappa\tau}$ $\Delta_{max}:=\max_{i,j}\|\Delta_{ij}\|_{M_j}$ with $\Delta_{ij}$ defined in \eqref{eq:Delta}, initial conditions $\varepsilon_0 := \|A_{d_0}^{-1}\boldsymbol{\beta}_{d_0} - P_{d_0}\bar X(0)\|_{M_{d_0}}$, and $\bar{\varepsilon}_0=||\bar{X}(0)||_{\bar M_{d_{1}}}$. The states $\bar{X}(t_i)$ at switching time instances are reset to a value according to the weighted least square method similar to \eqref{eq:X0_1}.
	\end{theorem}
\begin{proof}
	We show that the following inequalities hold with $\bar{\varepsilon}_i,\varepsilon_i$ satisfying \eqref{eq:difference_equations}:
	\begin{align*}
	&||\bar{X}(t_i)||_{\bar M_{d_{i+1}}}\le \bar{\varepsilon}_i \quad \text{ and }\quad 
	||X(t_i)-P_{d_i}\bar{X}(t_i)||_{M_{d_{i+1}}}\le \varepsilon_i.
	\end{align*}
Note that $\varepsilon_i$ and $\bar{\varepsilon}_i$ are defined inductively in \eqref{eq:difference_equations} and depend on each other. $\varepsilon_i$ bounds the norm of $X(t_i)-P_{d_i}\bar{X}(t_i)$ weighted by $M_{d_{i+1}}$ but $\bar{\varepsilon}_i$ bounds the norm of $\bar{X}(t_i)$ weighted by $\bar M_{d_{i+1}}$. In order to establish the relation between these two quantities inductively, we have to use the appropriate weight and change it using the definition $\bar M_d = P_d^T M_d P_d$ whenever necessary.

	At the $i^{\mathrm{th}}$ switching time instance, we have $X(t_i) = X(t_i^-)+\Delta_{i,i+1}$. By adding and subtracting the term $P_{d_i}\bar X(t_i^-)$ and noting that $M_{d_{i+1}}\leq\mu M_{d_i}$, we can write:
	\begin{align}
	\label{eq:mid_ineq1}
	\|X& (t_i^-) +\Delta_{i,i+1} - P_{d_{i+1}}\bar X(t_i)\|_{M_{d_{i+1}}}\nonumber\\
	& = \|X (t_i^-) - P_{d_i}\bar X(t_i^-) + P_{d_i}\bar X(t_i^-) +\Delta_{i,i+1} - P_{d_{i+1}}\bar X(t_i)\|_{M_{d_{i+1}}}\nonumber && \text{(by including $\pm P_{d_i}\bar X(t_i^-)$)}\\
	& \le  \|X(t_i^-) - P_{d_i}\bar X(t_i^-)\|_{M_{d_{i+1}}}+ \| P_{d_i}\bar X(t_i^-) + \Delta_{i,i+1} - P_{d_{i+1}}\bar X(t_i)\|_{M_{d_{i+1}}}\nonumber && \text{(by triangle inequality)}\\
	&\le  \mu \|X(t_i^-) - P_{d_i}\bar X(t_i^-)\|_{M_{d_{i}}} +
	\| P_{d_i}\bar X(t_i^-) + \Delta_{i,i+1} - P_{d_{i+1}}\bar X(t_i)\|_{M_{d_{i+1}}} && \text{(by using $M_{d_{i+1}}\leq\mu M_{d_i}$).}
	\end{align}
	For the time interval $[t_{i-1},t_i)$ we already know that
	\[
	\|X(t_i^-) - P_{d_i}\bar X(t_i^-)\|_{M_{d_{i}}} \leq \|X(t_{i-1}) - P_{d_i}\bar X(t_{i-1})\|_{M_{d_i}}e^{-\kappa_{d_i}(t_i-t_{i-1})}\leq g\varepsilon_{i-1},
	\]
	since the policy has dwell time $\tau$.
	Now we deal with the second term in \eqref{eq:mid_ineq1}.
	As a consequence of picking columns of $P_{d_i}\in \mathbb{R}^n\times\mathbb{R}^r$ from the corresponding unitary matrix, one can easily notice that $P_{d_i}^TP_{d_i}=\mathbb I_r$ and $P_{d_i}P_{d_i}^T\leq \mathbb I_m$ for every $i$. Therefore, using the triangle inequality we get
	\begin{align}
	\label{eq:mid2_ineq}
	\|P_{d_i}& \bar X(t_i^-) +\Delta_{i,i+1} - P_{d_{i+1}}\bar X(t_i)\|_{M_{d_{i+1}}}\nonumber\\
	& \le \|P_{d_i}\bar X(t_i^-)\|_{M_{d_{i+1}}}+\|\Delta_{i,i+1}\|_{M_{d_{i+1}}}+\|P_{d_{i+1}}\bar X(t_i)\|_{M_{d_{i+1}}}\nonumber\\
	& \le \mu\|\bar X(t_i^-)\|_{\bar M_{d_{i}}}+\Delta_{max}+\|\bar X(t_i)\|_{\bar M_{d_{i+1}}}.
	\end{align}
	The last inequality is due to $M_{d_{i+1}}\leq\mu M_{d_i}$, the definition of $\Delta_{max}$ in Theorem~\ref{thm:bound2}, and the definition $\bar M_d = P_d^T M_d P_d$ in Lemma~\ref{lem:bound}.
	$\bar X(t_i)$ is selected as the minimiser of the expression 
	\begin{equation}
	\|P_{d_{i}}\bar X(t_i^-) +\Delta_{i,i+1} - P_{d_{i+1}}\bar X(t_i)\|_{2},
	\end{equation}
	which is 
	\begin{equation}\label{eq:Xbar_reset}
		\bar X(t_i)= P_{d_{i+1}}^T(P_{d_{i}}\bar X(t_i^-)+\Delta_{i,i+1}).
	\end{equation}
		Therefore,
		\begin{align*}
		 \|\bar X(t_i)\|_{\bar M_{d_{i+1}}}^2& =(P_{d_i}X(t_i^-)+\Delta_{i,i+1})^T P_{d_{i+1}}P_{d_{i+1}}^T M_{d_{i+1}}P_{d_{i+1}}P_{d_{i+1}}^T(P_{d_i}X(t_i^-)+\Delta_{i,i+1})\\
		 & \le (P_{d_i}X(t_i^-)+\Delta_{i,i+1})^T M_{d_{i+1}} (P_{d_i}X(t_i^-)+\Delta_{i,i+1}).
		\end{align*}
	Based on \eqref{eq:bounded_X} and taking dwell time $\tau$ into account, we know that
	 \[
	 \|\bar X(t_i^-)\|_{\bar M_{d_{i}}}\leq \|\bar X(t_{i-1})\|_{\bar M_{d_i}}e^{-\kappa \tau}.
	 \]
Then,
	\begin{align}
	\label{eq:mid3_ineq}
	\|\bar X(t_i)\|_{M_{d_{i+1}}}\leq\mu\|\bar X(t_i^-)\|_{\bar M_{d_{i}}}+\Delta_{max}\leq \mu g\|\bar X(t_{i-1})\|_{\bar M_{d_{i}}}+\Delta_{max}
	\end{align}
		 Putting \eqref{eq:mid3_ineq} into \eqref{eq:mid2_ineq} we have:
	\begin{align}\label{eq:mid4_ineq}
	&\|\bar X(t_i^-) +\Delta_{i,i+1} - P_{d_{i+1}}\bar X(t_i)\|_{M_{d_{i}}}\leq 2\mu g\|\bar X(t_{i-1})\|_{M_{d_{i+1}}}+2\Delta_{max}= 2\bar{\varepsilon}_i+2\Delta_{max}.
	\end{align}
	Combining the two computed upper bounds, we get the difference equations~\eqref{eq:difference_equations}.
	\end{proof}
	
\begin{remark}
	(1) The precision of the bound in \eqref{eq:difference_equations} can be increased in two ways. First, the bound will be lower for policies with larger dwell time $\tau$ (smaller $g$). Second, if we increase the order of reduced system, $\varepsilon_0$ will become smaller.	
	(2) The gain $g$ solely depends on the CTMDP $\N$ and dwell time of policy $\pi$. In order to have a meaningful error bound, dwell time should satisfy
	$\tau>\frac{\log\mu}{\kappa}$.
	This condition is already true if we find a \emph{common Lyapunov function} for the CTMDP $\N$, 
	i.e., if there is one matrix $M$ independent of the decision vector $d$
	satisfying \eqref{LMI_MDP}. 
	In that case, $\mu=1$ and dwell time can take any positive value.
	\end{remark}

	\begin{corollary}
		\label{cor:convergence}
		The error $\varepsilon_i$ in \eqref{eq:difference_equations} converges to the constant value $\gamma\Delta_{max}$ for $\mu g < 1$, where
			\begin{equation}
			\gamma:=\frac{2-4\mu g}{(1-\mu g)^2}.
		\end{equation}
	\end{corollary}
	\begin{proof}
		We can rewrite \eqref{eq:difference_equations} into a discrete time state space representation as
		\begin{equation}
			\label{eq:discrete_SS}
			\begin{bmatrix}
				\varepsilon_i\\
				\bar \varepsilon_i
			\end{bmatrix}=
			\begin{bmatrix}
				\mu g & 2\mu g\\
				0 & \mu g
			\end{bmatrix}
			\begin{bmatrix}
				\varepsilon_{i-1}\\
				\bar \varepsilon_{i-1}
			\end{bmatrix}
			+\begin{bmatrix}
				2\\
				1
			\end{bmatrix}\Delta_{max},
		\end{equation} 	
	We consider \eqref{eq:discrete_SS} as a dynamical system in discrete time (index $i$ plays the role of time, which is discrete). Such a discrete-time dynamical system is asymptotically stable if all eigenvalues of its state matrix are in the unit circle. Since the state matrix of \eqref{eq:discrete_SS} is upper triangular, its eigenvalues are the same as the diagonal elements of the state matrix, which are both $\mu g$. Therefore, the system is asymptotically stable iff $\mu g< 1$.
Hence, we can compute the steady state value of $\varepsilon$ using the expression below:
			\begin{equation*}
			 	\lim_{i\rightarrow\infty}
				\varepsilon_i = 
				\begin{bmatrix}
			 		1&0
			 	\end{bmatrix}		
				\begin{bmatrix}
					1-\mu g & 2\mu g\\
					0 & 1-\mu g
				\end{bmatrix}^{-1}
				\begin{bmatrix}
					2\\
					1
				\end{bmatrix}\Delta_{max}=\frac{2-4\mu g}{(1-\mu g)^2}\Delta_{max}.
		\end{equation*}
	\end{proof}
\begin{remark}
	For the case of having no $\bad$ states, we get $A^{-1}\mathbf{\beta_d}=-\mathbf{1}$ and $\Delta_{max}=0$. Corollary~\ref{cor:convergence} implies that for CTMDP $\N$ with no $\bad$ states, the error bound will converge to zero as a function of time. 
\end{remark}

So far we discussed reduction and error computation for a given policy $\pi$. Our proposed CTMDP reduction scheme is outlined in Algorithm ~\ref{alg:CTMDP_red}. Notice that the statement of Theorem~\ref{thm:bound2} holds for any policy as long as it has a dwell time at least $\tau$.
Therefore, we can find a policy using a reduced system and apply it to the original CTMDP $\N$ with the goal of maximising reachability probability.
For a given CTMDP $\N$, time horizon $T$, probability threshold $\theta$, and error bound $\varepsilon$, we select a dwell time $\tau$ and order of the reduced system such that $\varepsilon_ne^{-\kappa(T-t_n)}\le \varepsilon$ according to \eqref{eq:bound_MDP2} with $n = \lfloor T/\tau\rfloor$. Then we construct a policy $\pi$ using the reduced order system~\eqref{eq:switch_red} by setting $d_0 = \arg\max_d{A_d X_d(0)}$ where $X_d(0)=A_d^{-1}\beta_d$. The next selection of policies are done by respecting dwell time and $d_{i+1} = \arg\max_d{P_{d}\bar A_d \bar X_d(t)}$ for $t\ge t_i+\tau$ with $t_i$ being the previous switching time. 
Policy synthesis over the reduced order system can be implemented as it is shown in Algorithm ~\ref{alg:CTMDP_policy}.
Note that the computed policy may not be optimal because we fix a dwell time and a discretisation time step.
If the computed interval for reachability probability is not above $\theta$, we go back and improve the results by increasing the order of the reduced system.

\begin{example}
		Consider a CTMDP described by the following generator matrices corresponding to two decisions $d_1$ and $d_2$,
		\begin{equation*}
		Q_{d_1}=
		\begin{bmatrix}
		-1 & 1 & 0 & 0 & 0\\
		0.01 & -3.01 & 0.5 & 0.5 & 2\\
		0 & 0.01 & -1.01 & 0 & 1\\
		0 & 0.01 & 0.05 & -1.06 & 1\\
		0 & 0 & 0 & 0 & 0
		\end{bmatrix},\,
		Q_{d_2}=
		\begin{bmatrix}
		-1.5 & 0 & 0.75 & 0.75 & 0\\
		0.01 & -3.01 & 0.5 & 0.5 & 2\\
		0 & 0.01 & -1.01 & 0 & 1\\
		0 & 0.01 & 0.05 & -1.06 & 1\\
		0 & 0 & 0 & 0 & 0
		\end{bmatrix}.
		\end{equation*}
		This means there are two actions available in the first state, and each induces outgoing rates specified by the first rows of $Q_{d_1}$ and $Q_{d_2}$. The other states have only one action available. The last state is $\good$, which is absorbing. We set the time bound $T=10$. Using the partition defined in Eq.~\eqref{eq:partition}, we get
		\begin{equation*}
		A_{d_1}=
		\begin{bmatrix}
		-1 & 1 & 0 & 0 \\
		0.01 & -3.01 & 0.5 & 0.5\\
		0 & 0.01 & -1.01 & 0\\
		0 & 0.01 & 0.05 & -1.06
		\end{bmatrix}\!,\! \boldsymbol{\beta}_{d_1}=
		\begin{bmatrix}
		0\\
		2\\
		1\\
		1
		\end{bmatrix}\!,A_{d_2}=
		\begin{bmatrix}
		-1.5 & 0 & 0.75 & 0.75\\
		0.01 & -3.01 & 0.5 & 0.5\\
		0 & 0.01 & -1.01 & 0\\
		0 & 0.01 & 0.05 & -1.06
		\end{bmatrix}\!,\!\boldsymbol{\beta}_{d_2}=
		\begin{bmatrix}
		0\\
		2\\
		1\\
		1
		\end{bmatrix}\!.
		\end{equation*}
		Both $A_{d_1}$ and $A_{d_2}$ are irreducible. Thus, Assumption~\ref{ass:stab_MDP} holds. We compute the decay rates $\kappa_{d_1}$ and $\kappa_{d_1}$ using  Eq.~\eqref{kappa2} and set $\kappa = \min (\kappa_{d_1},\kappa_{d_2}) =0.4965$. Furthermore, Eq.~\eqref{LMI_MDP} can be satisfied by setting $M_{d_1}=M_{d_2}=\mathbb{I}_4$. This allows us to choose $\mu=1$. Hence, the dwell time $\tau$ can take any positive value since $\frac{\log{\mu}}{\kappa}=0$. We set the dwell time $\tau=2.3$.\\
		For the reduced order $r=3$, we use Theorem~\ref{th:algo-reduction2} and get
		\begin{equation*}
		\bar A_{d_1}=
		\begin{bmatrix}
		-3.0199  & 0.9859 & 0.6244\\
		0  & -0.993  & -0.3137\\
		0    &     0  & -1.0071
		\end{bmatrix},\quad
		\quad P_{d_1}=
		\begin{bmatrix}
		-0.4437  & -0.8962  & -0.0059\\
		0.8962  & -0.4437 &   0.0041\\
		-0.0045   &-0.0024  &  0.7071\\
		-0.0045  & -0.0024   & 0.7071
		\end{bmatrix}
		\end{equation*}
		that correspond to the decision vector $d_1$, and
		\begin{equation*}
		\bar A_{d_2}=
		\begin{bmatrix}
		-3.0149  & -0.0174& 0.6982\\
		0  & -1.5  & -1.0571\\
		0    &     0  & -1.0049
		\end{bmatrix},\quad
		\quad P_{d_2}=
		\begin{bmatrix}
		0.0049  & -1  & -0.0001\\
		1  & 0.0049 &   0.0071\\
		-0.005   &0.0001  &  0.7071\\
		-0.005   &0.0001  &  0.7071
		\end{bmatrix}
		\end{equation*}
		that correspond to $d_2$. 
		We initialise the set of differential equations with $\bar X_{d_1}(0)$ and $\bar X_{d_2}(0)$ computed using Eq.~\eqref{eq:X0_1} as
		 \begin{equation*}\bar X_{d_1}(0)=
		 		\begin{bmatrix}
		 		-0.4436\\
		 		1.3447\\
		 		-1.4125
		 		\end{bmatrix}\quad {and}\quad\bar X_{d_2}(0)=
		 		\begin{bmatrix}
		 		-0.9949\\
		 		0.9949\\
		 		-1.4124
		 		\end{bmatrix}.
		 \end{equation*}
		 Note that $g=e^{-\kappa \tau}=0.007$, $\Delta_{12}=\Delta_{21}=0$, and $\Delta_{max}=0$. 
		 We compute the error of order reduction using equations \eqref{eq:bound_MDP2}-\eqref{eq:difference_equations} with $n=\lfloor \frac{T}{\tau}\rfloor=4$ and $t_n=n\tau=9.2$. This gives the error bound $0.1396$.		 
		\label{example_ctmdp}
	\end{example}

Our formulated error bound depends on the order $r$ of the reduced system and the dwell time $\tau$. There is a tradeoff between $r$ and $\tau$ for having a guaranteed error bound. The error bound depends on $r$ implicitly and is selected recursively. Computation of the sub-optimal policy depends also on the discretisation step $\delta$. The overall complexity of such a computation for a CTMDP with $m$ states, $l$ decision vectors, and time bound $T$ is $\mathcal{O}(lm^3) + \mathcal{O}(\frac{Tlr^2}{\delta})$, where the first and second terms are the computational complexities for the reduced system and the sub-optimal policy, respectively.

\begin{algorithm}	
	\caption{Order reduction of CTMDPs}
	\label{alg:CTMDP_red}
	\SetAlgoLined
	\KwIn{CTMDP $\N$, time bound $T$, maximum error bound $\varepsilon$, policy $\pi$ with dwell time $\tau$}		
	\begin{enumerate}
		\item Compute $A_d$, $\mathbf \beta_d$ and $\kappa_d$ for all $d$, based on \eqref{eq:partition} and \eqref{kappa2}
		\item Set $\kappa = \min_d\kappa_d$ and $M_d = \mathbb I_{|S_{\N}|}$
		\item Compute the maximum number of switches as $n=\lfloor\frac{T}{\tau} \rfloor$
		\item Initialise the order $r=0$\\
		\item \textbf{Do}\\
		\quad $r\leftarrow r+1$\\
		\quad Compute $\bar A_d$ and $P_d$ for all $d\in \mathbf D$ using \eqref{Schur}\\
		\quad Compute $\bar X_d(0)$ for all $d\in \mathbf D$ using \eqref{eq:X0_1}\\
		\quad Compute error bound $\varepsilon_r$ as \eqref{eq:bound_MDP2} using \eqref{eq:difference_equations}\\
		\textbf{While} ($\varepsilon_r\geq \varepsilon$)
	\end{enumerate}
	\KwOut{Reduced order system of \eqref{eq:switch_red} with matrices ($\bar A_d$, $P_d$ for $d\in \mathbf D$)}			
\end{algorithm}

\begin{algorithm}	
	\caption{Sub-optimal policy synthesis for CTMDPs}
	\label{alg:CTMDP_policy}
	\SetAlgoLined
	\KwIn{Reduced system ($\bar A_d$, $P_d$ for $d\in \mathbf D$), time bound $T$, dwell time $\tau$, discretisation step $\delta$}
	\begin{enumerate}
		\item $d_0  = \arg\max\limits_{d\in\mathbf D}(Q_dX_d(0))$
		\item $k =  \lfloor\frac{\tau}{\delta}\rfloor+1$\\
		\item $\pi(t)=d_0$ for $t\in [0,k\delta)$\\
		\item \textbf{While $k<\lfloor\frac{T}{\delta}\rfloor+1$}\\
		\quad Compute a possibly sub-optimal policy using:
		\begin{equation*}
		d_k=arg\max\limits_{d\in\mathbf D}(Q_dP_d\bar X_d(k\delta))
		\end{equation*}
		\quad \textbf{If $d_k\neq d_{k-1}$}\\
		\qquad $\pi(t)=d_k$ for $t\in [k\delta,(k+\lfloor\frac{\tau}{\delta}\rfloor+1)\delta$)\\	
		\qquad $k\leftarrow k+\lfloor\frac{\tau}{\delta}\rfloor+1$\\
		\qquad Compute $\bar X_d(k\delta)$ using \eqref{eq:switch_red} and \eqref{eq:Xbar_reset} for all $d\in \mathbf{D}$\\
		\quad \textbf{Else}\\
		\qquad $\pi(t)=d_k$ for $t\in [k\delta,(k+1)\delta)$\\
		\qquad $k\leftarrow k+1$\\
		\qquad Compute $\bar X_d(k\delta)$ using \eqref{eq:switch_red} for all $d\in \mathbf{D}$\\
		\quad \textbf{End}\\		
		\textbf{End}\\
		\end{enumerate}
	\KwOut{Sub-optimal policy $\pi(t)$ for $t\in [0,T]$}
\end{algorithm}

\section{Simulation Results}
\label{sec:Simulation}
In this section, we first use our method for reachability analysis of two queuing systems, namely $M/M/1$ and tandem networks. We then evaluate the performance of our proposed symbolic computation on randomly generated models.

The $M/M/1$ queue consists of only one queue with a specific capacity denoted by $\textbf{\textsf{cap}}$. Jobs arrive with the rate $\bar{\lambda}$ and are processed with the rate $\mu$. The $M/M/1$ queue can be modelled as a CTMC with a state space of size $(\textbf{\textsf{cap}}+1)$. We find the probability of reaching the configuration in which the queue is at its full capacity from a configuration in which the queue is empty. The generator matrix of this CTMC is tridiagonal, with upper diagonal entries $\bar{\lambda}$, lower diagonal entries $\mu$, and main diagonal entries $-(\bar{\lambda}+\mu)$.

We choose $\textbf{\textsf{cap}}=100$ (size of the state space is $101$) and fix the size of the reduced system to $r=10$. We also fix the arrival rate $\bar \lambda=10$ and study the behaviour of our formulated error bound for state reduction with respect to the processing rate $\mu$.
 Fig.~\ref{fig:MM1} (\textbf{\textsf{left}}) demonstrates the variations of the decay rate $\kappa$ defined in Eq.~\eqref{kappa2} as a function of processing rate $\mu$. The decay rate is larger for smaller values of $\mu$ and become very close to zero for larger values of $\mu$, which makes our approach very efficient for smaller values of $\mu$. This fact is also visible from Fig.~\ref{fig:MM1} (\textbf{\textbf{right}}), where the error defined formally in Eq.~\eqref{eq:error} is shown as a function of the time bound $T$ and $\mu$ in logarithmic scale.
 It can be observed that the error is very small for larger time bounds $T$ and smaller $\mu$.
\begin{figure}[t]
\begin{center}
\includegraphics[width=0.47\textwidth]{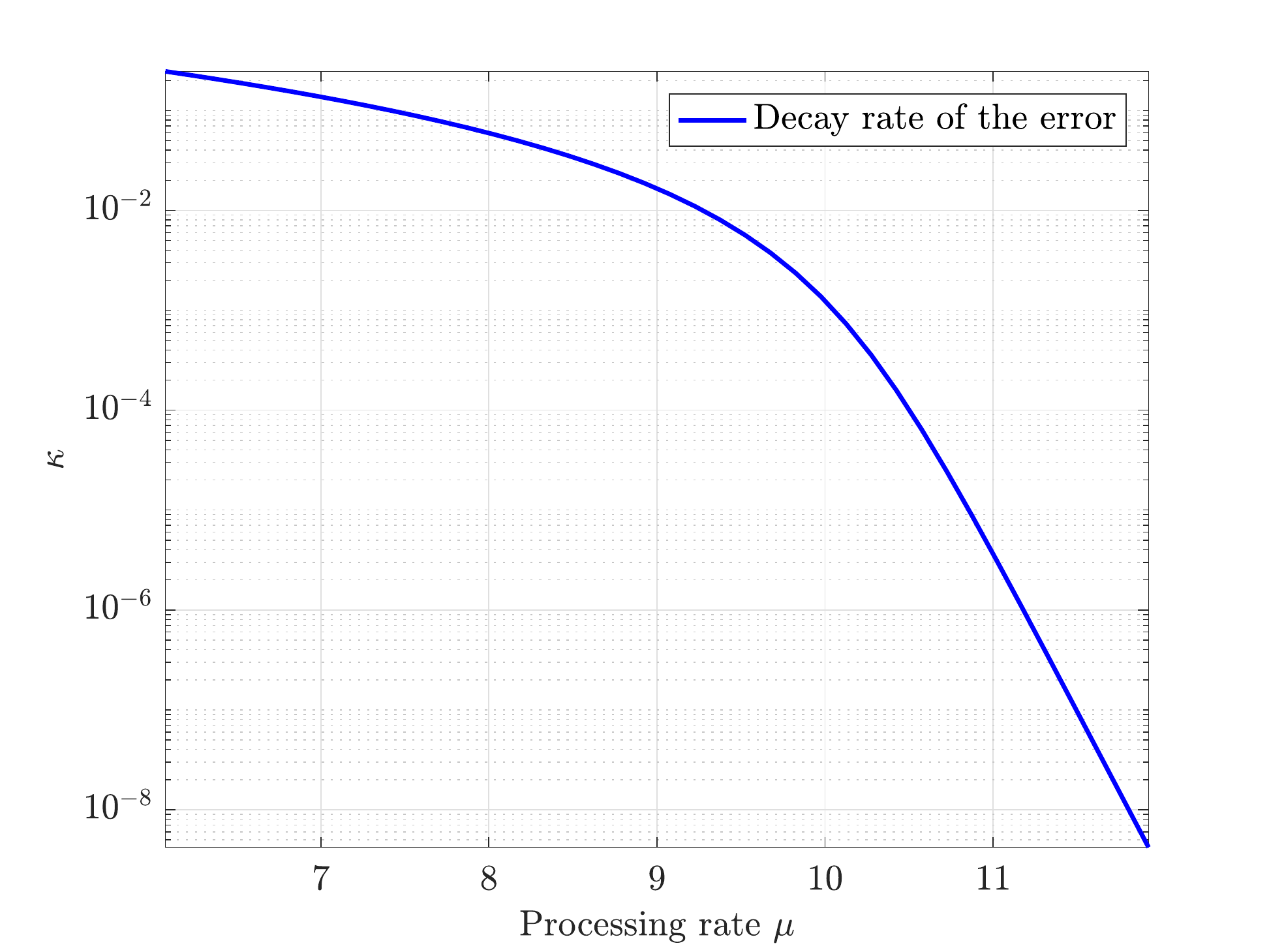} 
\includegraphics[width=0.47\textwidth]{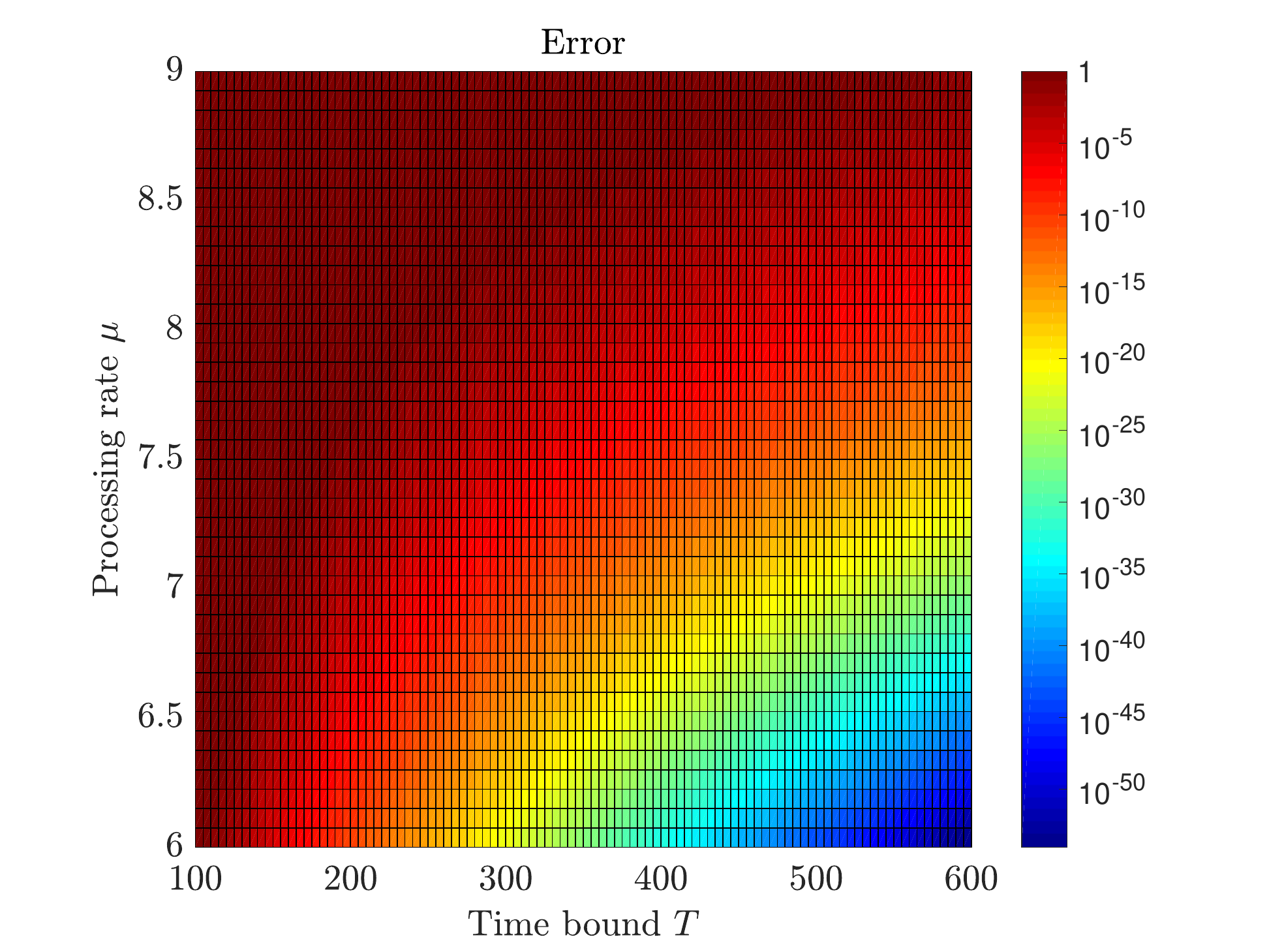} 
\caption{Error analysis for the state reduction for $M/M/1$ queuing system.  \textbf{\textsf{left:}} decay rate of the error as a function of processing rate $\mu$. \textbf{\textsf{right:}} error of the state reduction as a function of time bound $T$ and processing rate $\mu$. The error is very small for larger time bounds $T$ and smaller $\mu$.} 
\label{fig:MM1}
\end{center}
\end{figure}

We now apply our results to the \emph{tandem network} shown in Fig.~\ref{fig:tandem_net}. The network is a queuing system that consists of a $M/Cox2/1$ queue composed with a $M/M/1$ queue \cite{Holger:1999}.
\begin{figure}[t]
\begin{center}
\includegraphics[width=11cm]{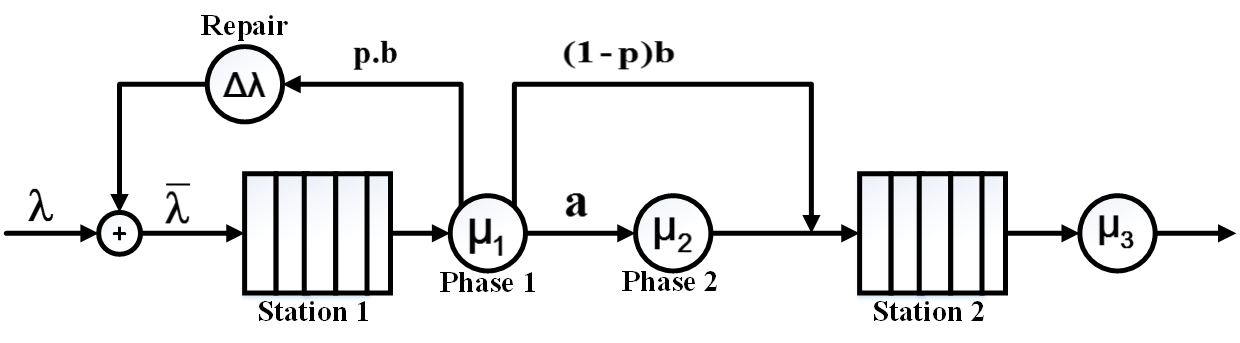} 
\caption{A typical tandem network} 
\label{fig:tandem_net}
\end{center}
\end{figure}
Both queuing stations have a capacity of $\textbf{\textsf{cap}}$. The first queuing station has two phases for processing jobs while the second queuing station has only one phase. Processing phases are indicated by circles in Fig.~\ref{fig:tandem_net}.
Jobs arrive at the first queuing station with rate $\bar{\lambda}$ and are processed in the first phase with rate $\mu_1$. After this phase, jobs are passed through the second phase with probability $a$, which are then processed with rate $\mu_2$. Alternatively, jobs will be sent directly to the second queuing station with probability $b$, a percent of which will have to undergo a repair phase and will go back to the first station with rate $\Delta\lambda$ to be processed again.
This percentage is denoted by $p$.
Processing in the second station has rate $\mu_3$.

The tandem network can be modelled as a CTMC with a state space of size determined by $\textbf{\textsf{cap}}$.
We find the probability of reaching to the configurations in which both stations are at their full capacity (blocked state) starting from a configuration in which both stations are empty (empty state).
We consider $\textbf{\textsf{cap}} = 5$ which results in a CTMC with $65$ states.
We have chosen values $\mu_1 = \mu_2 = 2$, $\mu_3 = \lambda = 4$, $a=0.1$, and $b = 0.9$. We also set $p=0$ and $\Delta\lambda = 0$, which means no job is going to the repair phase.
Matrix inequalities~\eqref{LMI0} are satisfied with $M$ being identity and $\kappa = 0.001$.
Using the reduction technique of Section~\ref{sec:CTMCs}, we can find approximate solution of reachability with only $3$ state variables.
Fig.~\ref{fig:merged} (\textbf{\textsf{left}}) shows reachability probability computed over the tandem network and the reduced order system together with the error bound as a function of time horizon.
The error has the initial value $0.02$, computed via the choice of initial reduced state in \eqref{eq:X0_1}, and converges to zero exponentially with rate $ 0.0013$. It can also be noticed that the outputs of the full and reduced-order systems cannot be distinguished in the figure. This is due to the fact that  their actual difference is very small compared to the formal error bound characterised in this paper.

Fig.~\ref{fig:merged} (\textbf{\textsf{right}}) gives the error bound as a function of time horizon of reachability and order of the reduced system. As discussed, the error goes to zero exponentially as a function of time horizon. It also converges to zero by increasing the order of reduced system.\\
Now consider a scenario that the network can operate in \emph{fast} or \emph{safe} modes. In fast mode, fewer jobs are sent through the second phase (corresponding to a smaller value of $a$); this, in turn, increases the probability that jobs which did not pass second phase, need to be processed again. We model influence of returned jobs as an increase in $\Delta\lambda$.\\
We consider the case that there are two possible rates $a\in\{0.6, 0.7\}$ corresponding respectively to fast and safe modes.
If fast mode is chosen, $10\%$ of jobs will be returned ($p=0.1$) with rate $\Delta\lambda=0.05$. In the safe mode, only $5\%$ of jobs ($p=0.05$) will be returned with the same rate $\Delta\lambda$.
We set $\mu_1 = \mu_2 = 2.5$ and $\mu_3 = \lambda = 3$.

A tandem network with capacity $\textbf{\textsf{cap}} = 2$ and these two modes can be modelled as a CTMDP with $16$ states and $16$ decision vectors. Fig.~\ref{fig4} depicts state diagram of this CTMDP with states $S_1,S_2,S_3,S_4$ having two modes with the corresponding value of rate $a$.
 We assume the tandem network is initially at the state $220$ of Fig.~\ref{fig4}, which means there are two jobs in the first station, both are being served in the second phase, and there is no job in the second station. We consider synthesising a strategy with respect to the probability of having both queuing stations becoming empty by time $T$.
We have implemented the approach of Section~\ref{sec:CTMDPs} and obtained a reduced system of order $6$ with $\varepsilon_0=0.14$.
Fig.~\ref{fig:dwell_errors} (\textbf{\textsf{left}}) demonstrates reachability probabilities as a function of time for both the tandem network and its reduced counterpart together with the error bound.
Intuitively, choosing the fast mode in the beginning will result in faster progress of the tasks, especially when queues are more loaded; however, 
if this selection is continued, it will result in a high number of returned jobs, which is not desired.
This behaviour is observed depending on the state in the form of three switches in states $S_2,S_3,S_4$.
In Fig.~\ref{fig:dwell_errors} (\textbf{\textsf{left}}) the green trajectory corresponds to the reachability probability of the original CTMDP under the non-restricted optimal piecewise constant policy. 
Fig.~\ref{fig:dwell_errors} (\textbf{\textsf{right}}) demonstrates the impact of dwell time on the optimisation error (in blue) and on the guaranteed error bound (in red) for time bound $T = 100$ seconds.
The reduction error bounds are computed formally using the results of Theorem~\ref{thm:bound2}, by solving \eqref{eq:difference_equations} and using it in \eqref{eq:bound_MDP2}. The optimisation error is computed numerically. For each dwell time, we compute optimal reachability probability corresponding to the full-order system running with non-restricted policy as well as the reachability probability corresponding to the reduced-order system with policy restricted with the chosen dwell time. The optimisation error is defined as the difference between these two values.

\begin{figure}[t]
	\includegraphics[width = \textwidth]{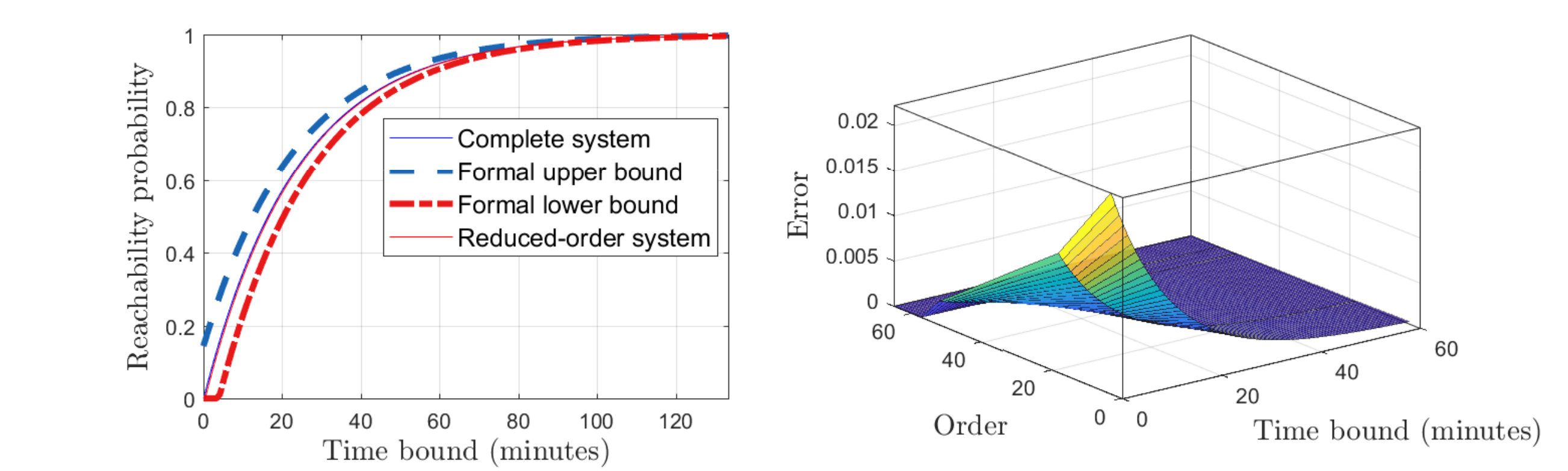} 
	\caption{
	\textbf{\textsf{left:}} approximate reachability probability for tandem network as a function of time horizon with guaranteed error bounds.
	\textbf{\textsf{right:}} error bound as a function of time horizon and order of the reduced system;
	}
	\label{fig:merged}
\end{figure}

\begin{figure}[t]
\includegraphics[width=0.8\textwidth]{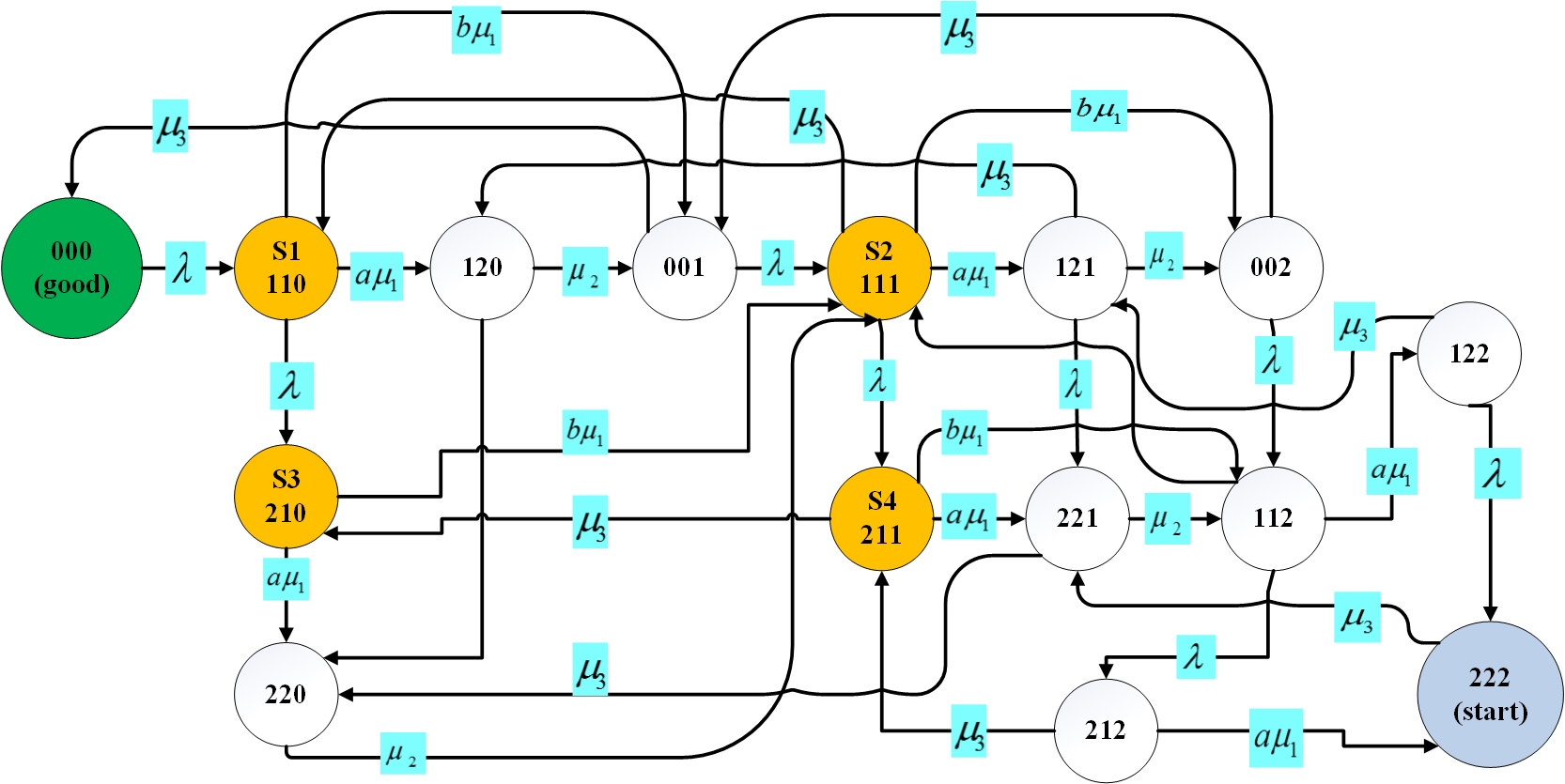} 
\caption{State diagram of a CTMDP with $16$ states and $16$ decision vectors corresponding to a tandem network with capacity $2$. States $S_1,S_2,S_3,S_4$ have two modes with rates $a\in\{0.6,0.7\}$.}
\label{fig4}
\end{figure}
\begin{figure}[t]
	\includegraphics[width = \textwidth]{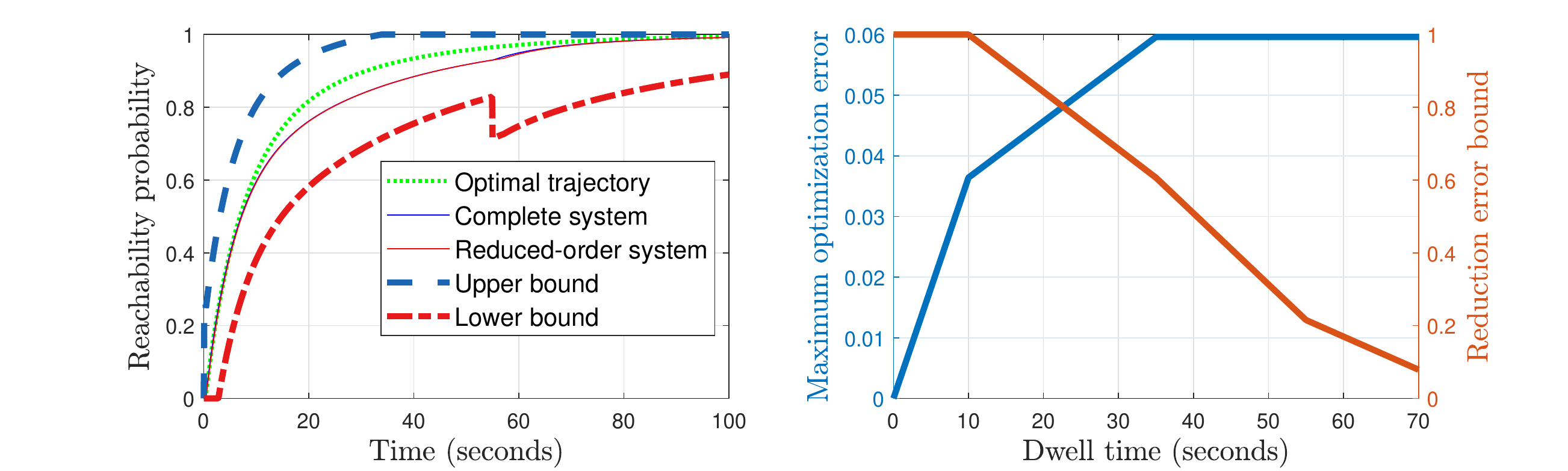} 
	\caption{\textbf{\textsf{left:}} approximate reachability probability for tandem network with $16$ decision vectors including and the formal bounds ($\tau=55$ seconds). \textbf{\textsf{right:}} error in the optimal reachability probability and  the reduction error bound with dwell time ($T=100$ seconds).} 
	\label{fig:dwell_errors}
\end{figure}
Finally, we assess the performance of symbolic computation on randomly generated models. Table~\ref{tb:symbolic_runtime} compares runtime of the reachability probability computation using three different methods:
adaptive implementation of the uniformisation technique presented in \cite{Buch:2011} ($RT_u$),
symbolic computation presented in our work without state reduction ($RT_s$) using only Algorithm~\ref{alg:CTMDP_policy} of \autoref{sec:app0}, and symbolic computation with state reduction ($RT_{sr}$) by running both Algorithms~\ref{alg:CTMDP_red}~and~\ref{alg:CTMDP_policy}.

Note that the method presented in \cite{Buch:2011} is developed for sub-optimal policy synthesis of CTMDPs and tunes the length of the time discretisation adaptively. According to our experiments, the adaptive selection of time discretisation makes it more efficient also for reachability computation of CTMCs in comparison with the uniform discretisation proposed in \cite{Baier:2003}. Therefore, we compare our results with the approach of \cite{Buch:2011}.

The experiments are done using MATLAB R2017a on a $3.3$ GHz Intel Core i5 processor. For each experiment, $10$ stochastic matrices are generated randomly as infinitesimal generator matrix corresponding to a CTMC without imposing any sparsity assumption. To implement the uniformisation, the step time is tuned adaptively with maximum truncation error bound $0.01$. The maximum number of terms in the Maclaurin expansion is set to $5$ and the time bound is fixed at $5$ seconds, while the minimum time step for uniformisation is chosen to be $10^{-4}$ seconds.
Note that $RT_{sr}$ also includes the time for running Algorithm~\ref{alg:CTMDP_red}. As it can be observed from Table~\ref{tb:symbolic_runtime}, $RT_{sr}$ is smaller than $RT_{u}$ and $RT_{s}$ by at least two and one orders of magnitude, respectively.

\begin{table}[t]
	\begin{center}
		\caption{Comparison of runtime (in seconds) for the reachability probability computation using
		the uniformisation technique of \cite{Buch:2011} ($RT_u$),
symbolic computation without state reduction ($RT_s$) by running only Algorithm~\ref{alg:CTMDP_policy}, and symbolic computation with state reduction ($RT_{sr}$) by running both Algorithms~\ref{alg:CTMDP_red}~and~\ref{alg:CTMDP_policy}.}
\label{tb:symbolic_runtime}
		\begin{tabular}{cccc}
			Number of states & $RT_u$ & $RT_s$ & $RT_{sr}$\\
			\hline
			100 & 3.132 & 0.0781& 0.0112\\
			200 & 7.295 & 0.5483& 0.0362\\
			500 & 94.55 & 8.247& 0.2371\\
			800 & 461.8 & 35.31& 0.9968\\
			1000 & 831.8 & 68.61& 1.788\\
			1200 & 1444.2 & 114.73& 2.4911\\
			1500 & 3384.1 & 226.21& 4.8538\\ \hline
		\end{tabular}	
	\end{center}
\end{table}

\section{Discussions}
\label{sec:Discussions}

We have taken a control-theoretic view on the time bounded reachability
problem for CTMCs and CTMDPs.
We show the dynamics associated with the problems are stable,
and use this as the basis for state space reduction.
We define reductions as generalised projections between state spaces
and find a Lyapunov characterisation of the error between the original
and the reduced dynamics.
This provides a formal error bound on the solution which 
decreases exponentially over time.
Our experiments on queueing systems demonstrate that, as the time horizon
grows, we can get significant reductions in state (and thus, model checking complexity).

We formulated a set of matrix (in)equalities that characterises the reduced-order system of equations.
We also provided algorithms for computing a feasible solution of these (in)equalities. For CTMDPs, our algorithm provides the error bound between the original and the reduced-order systems for the synthesised policy, but does not provide any result related to the optimality of the policy. Future directions of this work include combining this technique with the results in the literature to find a sub-optimal policy with guaranteed error bounds. Our algorithm assumes that the CTMDP is irreducible for any given decision vector. Finding ways to relax this assumption will increase its applicability.   
In the present paper, we have laid the theoretical foundations for the approach. We leave the comprehensive benchmarking of the approach for a separate publication.

\bibliographystyle{ACM-Reference-Format}
\bibliography{references}

\appendix
\label{sec:appendix}

\section{Error Bounds for $\varepsilon$-Bisimilar CTMCs}
\label{sec:app1}
Given matrices $A$ and $\bar A$ corresponding to stochastic matrices $Q$ and $\bar Q$, suppose that there exists a matrix $P_b$ such that $AP_b = P_b \bar A + \Delta A P_b$ and $\boldsymbol{\beta} = P_b\bar{\boldsymbol{\beta}} + \Delta \boldsymbol{\beta}$, where all elements of $\Delta A$ and $\Delta \boldsymbol{\beta}$ are bounded by $\varepsilon$ in the absolute value sense. Hence, a CTMC with $\hat A = A - \Delta A$ and $\hat{\boldsymbol{\beta}} = \boldsymbol{\beta} - \Delta\boldsymbol{\beta}$ can be reduced based on the notion of exact bisimulation. $\Delta{A}$ and $\Delta\boldsymbol{\beta}$ include all rate mismatches with respect to the equivalence classes specified by $P_b$. Defining the error vector as $e(t)=X(t)-P_b\bar{X}(t)$, dynamics of error would be as the following:
\begin{align}
\label{eq:eps_err_eq}
&\dot{e}(t)=Ae(t)+\Delta AP_b\bar{X}(t)\nonumber\\
&\dot{\bar{X}}(t)=\bar{A}\bar{X}(t)
\end{align}
Since $A$ and $\bar{A}$ are both stable matrices (extracted from the stochastic matrices $Q$ and $\bar Q$), steady state value of the vector $e(t)$ would be zero. The next theorem gives a bound on $e(t)$ for the case that absolute value of elements of $\Delta A$ and $\Delta \boldsymbol{\beta}$ do not exceed a certain threshold $\varepsilon$.
\begin{theorem}
	Suppose that elements of $\Delta A$ and $\Delta \boldsymbol{\beta}$ are bounded by $\varepsilon$. The elements of the error $e(t)\in \mathcal R^m$ defined in \eqref{eq:eps_err_eq} are bounded by
	\begin{equation*}
	|e_i(t)|\leq (m\varepsilon+\rho) \varLambda_i
	\end{equation*}  
	where, $\rho=||e(0)||_{\infty}$, $\varLambda=-{A}^{-1}$ and $\varLambda_i=\sum_{j=1}^{m}\varLambda(i,j)$.
\end{theorem}
\begin{proof}
	Let us denote state transition matrix $G(t):=e^{{A}t}$ and write its $i^{\text{th}}$ row as  $g_i(t)$. We also denote the $i^{th}$ column of $\Delta A$ by $\Delta A_i$.
	For $e_i(t)$ we can write:
	\begin{align*}
	e_i(t)&=g_i(t)\Delta A \ast P_b\bar{X}(t)+g_{i}(t)e(0)
	=\sum_{j=1}^{m}\int_0^t g_i(t-\tau)\Delta A_jF_j(\tau) d\tau +g_{i}(t)e(0)
	\end{align*}    
	where, $\ast$ operator stands for convolution of two signals in time domain and $F_j(t)$ is a scalar and obtained by multiplying $j^{\text{th}}$ row of $P_b$ by vector $\bar{X}(t)$ which is bounded by $1$. Therefore:
	\begin{align*}
	|e_i(t)| 
	&\leq \varepsilon n \int_0^t ||g_i(\tau)||_1d\tau +g_{i}(t)e(0)
	\end{align*}
	Moreover, for every arbitrary time $t\geq0$ we have $\|e_i(t)\|\leq(\varepsilon m+\rho)\int_0^{\infty} ||g_i(\tau)||_1d\tau $. However, this bound cannot be easily found since it requires computing $G(t) = e^{At}$. To avoid the computation of $G(t)$, we use the uniformised form of $Q$ defined as $H_0:=\frac{Q}{\gamma_0}+\mathbb I_{m+2}$.
	$H_0$ is a row stochastic matrix and $\gamma_0$ is the maximum of absolute value of diagonal elements of $Q$. Using $H_0$ one can compute state transition matrix corresponding to $Q$ as \cite{Buch:2011}:
	\begin{equation*}
	e^{Qt}=\sum_{k=0}^{\infty}H_0^ke^{(-\gamma_0 t)}\frac{(\gamma_0 t)^k}{k!}
	\end{equation*}
	It is easy to notice that for every $k$, inner argument in the above summation is (element-wise) non-negative. We can also expand $e^{Qt}$ in the following form:
	\begin{equation*}
	e^{\mathbf{Q}t} =
	\begin{bmatrix}
	e^{{A}t} & \vdots & (e^{{A}t}-I){A}^{-1}\boldsymbol{\beta}\cdots\\
	\dots & \qquad & \dots\\
	\mathbf{0} &\vdots & 1
	\end{bmatrix}
	\end{equation*}
	It can be seen that $e^{{A}t}$ is one of the blocks inside $e^{Qt}$. Therefore, $e^{{A}t}$ is (element-wise) a non-negative matrix for all $t\geq 0$. Using the definition of the Fourier transform of a function \cite{Ogata:2001}, we get
	\begin{equation*}
	\int_{0}^{\infty} |G_{ij}(\tau)|d\tau=\int_{0}^{\infty} G_{ij}(\tau)d\tau=-{A}^{-1}_{ij}
	\end{equation*}
	where, ${A}^{-1}_{ij}$ denotes the $ij^{th}$ element of $A^{-1}$. Setting $\varLambda:=-{A}^{-1}$ and $\varLambda_i:=\sum_{j=1}^{n}\varLambda(i,j)$, we get
	\begin{equation*}
	|e_i(t)|\leq (\varepsilon m+\rho) \varLambda_i.
	\end{equation*}		
\end{proof}

\section{Reducible CTMC case}
\label{sec:CTMC Red}
Throughout the paper, irreducibility of models is assumed. In this section, we show that our results are applicable to reducible CTMCs. The only assumption required for validity of the results of Section~\ref{sec:CTMCs} is the stability of the matrix $A$. We prove in the sequel that this assumption holds also for reducible CTMCs by preprocessing its structure and eliminating bottom strongly connected components (BSCCs) that do not affect the reachability probability.
\begin{remark}
	For any given time bound, the reachability probabilities corresponding to the BSCCs of the CTMC $\M$  are zero except for the BSCC containing the single state $\good$. Therefore, these BSCCs can be eliminated from the generator matrix. Thus we obtain a dynamical system for which the only BSCC is $\{\good\}$.
\end{remark}
\begin{proposition}
	For a reducible CTMC $\M$, after eliminating all the BSCCs except $\{\good\}$ and the states that can never reach $\{\good\}$, the matrix $A$ in \eqref{eq:Q_part} will be stable.
\end{proposition}
\begin{proof}
	If the CTMC is reducible, we first eliminate all the BSCCs except $\{\good\}$. We also eliminate states that can never reach $\{\good\}$. Therefore, the modified CTMC consists of only transient states and $\{\good\}$. The transient states can be partitioned into strongly connected components. The canonical form of matrix $A$ for such a CTMC will have the following structure:
	\begin{equation}
	A'=
	\begin{bmatrix}
	A'_{11} & A'_{12} & A'_{13}&\cdots&\cdots & A'_{1n}\\
	0 & A'_{22} &A'_{23}&A'_{24}&\cdots & A'_{2n}\\
	0 & 0 &A'_{33} &A'_{34}&\cdots & A'_{3n}\\
	\vdots & \vdots &\vdots&\vdots &\ddots&\vdots\\
	0 & 0 & 0 & \cdots & A'_{(n-1)(n-1)}&A'_{(n-1)n}\\
	0 & 0 & 0 & \cdots &0& A'_{nn},\\
	\end{bmatrix}\label{eq:canonical_reducible}
	\end{equation}
	where $A'_{ii}$s correspond to different strongly connected components. Since it is possible to reach from any state to $\{\good\}$, $A'_{ii}$s satisfy Assumption~\ref{ass:invert} are stable.
\end{proof}
Equation~\eqref{eq:Xp} with the block upper-diagonal matrix $A'$ in \eqref{eq:canonical_reducible} can be solved bottom-up while the order reduction can be utilised in each step.

\end{document}